\documentclass[11pt, a4paper]{imsart}

\usepackage[margin=2.5cm]{geometry}
\RequirePackage{amsthm,amsmath,amsfonts,amssymb}
\RequirePackage[numbers]{natbib}
\usepackage{amsthm}
\usepackage[boxed, linesnumbered, noend, noline]{algorithm2e}
\usepackage{aliascnt}
\usepackage{hyperref}

\usepackage{cleveref}
\let\oldtheorem\newtheorem
\RenewDocumentCommand{\newtheorem}{s m o m O{}}{%
\IfBooleanTF{#1}%
{\oldtheorem{#2}{#4}}%
{\IfNoValueTF{#3}{\oldtheorem{#2}{#4}[#5]}%
{\newaliascnt{#2}{#3}%
\oldtheorem{#2}[#2]{#4}%
\aliascntresetthe{#2}}}}

\startlocaldefs
\theoremstyle{plain}
\newtheorem{axiom}{Axiom}
\newtheorem{claim}[axiom]{Claim}
\newtheorem{theorem}{Theorem}[section]
\newtheorem{lemma}[theorem]{Lemma}
\newtheorem{corollary}[theorem]{Corollary}

\theoremstyle{remark}
\newtheorem{definition}[theorem]{Definition}

\newtheorem{remark}[theorem]{Remark}

\input{insbox}

\usepackage{pgfplots}
\usepackage{tikz}
\usepackage{graphicx}
\usepackage{subcaption}
\usetikzlibrary{shapes,decorations,arrows,calc,arrows.meta,fit,positioning}
\usepgfplotslibrary{fillbetween}
\usepackage{graphicx}
\usetikzlibrary{decorations.pathreplacing, matrix}
\usepackage{tabularx}
\usepackage[figuresright]{rotating}
\usepackage{makecell}
\usepackage{mathrsfs}

\SetKwInput{KwData}{Input}
\SetKwInput{KwResult}{Output}

\crefname{claim}{claim}{claims}
\crefname{lemma}{lemma}{lemmas}
\crefname{axiom}{axiom}{axioms}
\crefname{algorithm}{algorithm}{algorithms}

\usepackage{tikz}
\usetikzlibrary{shapes,decorations,arrows,calc,arrows.meta,fit,positioning}
\usetikzlibrary{shadows,shadings,shapes.symbols, patterns}
\tikzset{
    double color fill/.code 2 args={
        \pgfdeclareverticalshading[%
            tikz@axis@top,tikz@axis@middle,tikz@axis@bottom%
        ]{diagonalfill}{100bp}{%
            color(0bp)=(tikz@axis@bottom);
            color(50bp)=(tikz@axis@bottom);
            color(50bp)=(tikz@axis@middle);
            color(50bp)=(tikz@axis@top);
            color(100bp)=(tikz@axis@top)
        }
        \tikzset{shade, left color=#1, right color=#2, shading=diagonalfill}
    }
}

\usepackage{thmtools}
\usepackage{thm-restate}

\newcommand\SPIV{{\tt SPIV}}
\newcommand\SPEX{{\tt SPEX}}

\newcommand\COMP{{\tt COMP}}
\newcommand\SCIENT{{\tt SCIENT}}

\newcommand{\vX}{\vec X}

\renewcommand{\epsilon}{\eps}

\newcommand\vY{\vec Y}

\newcommand\vm{\vec m}
\newcommand\vp{\vec p}

\def\vGamma{\vec{\Gamma}}

\newcommand\vW{\vec W}
\newcommand\vS{\vec S}

\renewcommand{\vec}[1]{\boldsymbol{#1}}

\newcommand\KL[2]{D_{\mathrm{KL}}\bc{{{#1}\|{#2}}}}

\newcommand\SIGMA{\vec\sigma}

\newcommand\fE{\mathfrak{E}}

\newcommand\cA{\mathcal{A}}

\newcommand\cF{\mathcal{F}}
\newcommand\G{\mathcal{G}}
\newcommand\cE{\mathcal{E}}

\newcommand\cN{\mathcal{N}}

\newcommand\cT{\mathcal{T}}

\newcommand\cM{\mathcal{M}}

\def\cE{{\mathcal E}}

\newcommand\vZ{\vec Z}

\newcommand\vk{\vec k}

\newcommand\eul{\mathrm{e}}
\newcommand\eps{\varepsilon}

\newcommand\NN{\mathbb{N}}

\newcommand\Erw{\mathbb{E}}
\newcommand{\vecone}{\vec{1}}

\newcommand{\Po}{{\rm Po}}
\newcommand{\Bin}{{\rm Bin}}
\newcommand{\Mult}{{\rm Mult}}

\newcommand\bc[1]{\left({#1}\right)}

\newcommand\cbc[1]{\left\{{#1}\right\}}
\newcommand\bcfr[2]{\bc{\frac{#1}{#2}}}

\newcommand\brk[1]{\left\lbrack{#1}\right\rbrack}

\newcommand\norm[1]{\left\|{#1}\right\|}
\newcommand\abs[1]{\left|{#1}\right|}

\newcommand\RR{\mathbb{R}}

\newcommand{\Whp}{W.h.p.}
\newcommand{\whp}{w.h.p.}

\newcommand\pr{\mathbb{P}} 
\renewcommand\Pr{\pr}

\DeclareMathOperator*{\argmin}{arg\,min}

\newcommand{\ceil}[1]{\left\lceil#1\right\rceil}

 \def\G{{\vec G}}

\def\pr{{\mathbb P}}

\newcommand\SPOT{{\tt SPOT}}

\newcommand{\remove}[1]{}

\newcommand{\one}{V_1}

\newcommand\minf{m_{\mathrm{sc}}}
\newcommand\cinf{c_{\mathrm{sc}}}
\newcommand\mseed{m_{\mathrm{seed}}}
\newcommand\cseed{c_{\mathrm{seed}}}

\pgfplotsset{compat=1.14}

\usepackage{float}

\usepackage[colorinlistoftodos,prependcaption,textsize=small]{todonotes}
\usepackage{xcolor}

\newcommand{\invisible}[1]{}
\newcommand{\sss}{\scriptscriptstyle}
\endlocaldefs

\begin{document}

\begin{frontmatter}
\title{Algorithms for Threshold Group Testing}

\begin{aug}

\author[B]{\fnms{Amin}~\snm{Coja-Oghlan}\ead[label=e2a]{amin.coja-oghlan@tu-dortmund.de}\orcid{0000-0003-1331-9697}}
\author[A]{\fnms{Remco}~\snm{van der Hofstad}\ead[label=e1a]{r.w.v.d.hofstad@tue.nl}\orcid{0000-0003-1331-9697}}
\author[B]{\fnms{Lena}~\snm{Krieg}\ead[label=e2b]{lena.krieg@tu-dortmund.de}\orcid{0000-0002-5302-6955}}
\author[A]{\fnms{Noela}~\snm{Müller}\ead[label=e1b]{n.s.muller@tue.nl}\orcid{0009-0004-8182-311X}}
\author[A]{\fnms{Connor}~\snm{Riddlesden}\ead[label=e1c]{c.d.riddlesden@tue.nl}\orcid{0009-0003-4516-7349}}
\and
\author[B]{\fnms{Olga}~\snm{Scheftelowitsch}\ead[label=e2c]{olga.scheftelowitsch@tu-dortmund.de}}


\address[B]{TU Dortmund, Faculty of Computer Science\printead[presep={,\ }]{e2a,e2b,e2c}}
\address[A]{Eindhoven University of Technology, Department of Mathematics and Computer Science\printead[presep={,\ }]{e1a,e1b,e1c}}
\end{aug}

\begin{abstract}
We study the Threshold Group Testing (TGT) problem without a gap in the noiseless, non-adaptive setting, where the goal is to exactly recover a sparse binary vector from pooled test outcomes using as few tests as possible. In TGT, a test applied to a subset of items returns a positive outcome if the number of defective items in the subset reaches a prescribed threshold, and a negative outcome otherwise. Under the assumption of an analytic condition,  TGT has been shown to undergo a sharp information-theoretic phase transition for exact recovery on the class of constant-column test designs. In this paper, we develop an efficient inference algorithm that achieves exact recovery with high probability using the minimum number of non-adaptive tests that are needed for the constant-column design, thereby matching the information-theoretic threshold of a natural benchmark test design. Our approach is based on a spatially coupled test design and admits a significantly simpler analysis than existing algorithms for related group testing problems. In particular, unlike previous methods for binary group testing, our algorithm does not rely on the analysis of intricate weighted sums. This leads to a more straightforward proof technique, while still allowing near-optimal performance guarantees.
\end{abstract}



\end{frontmatter}


\tableofcontents
\section{Introduction}

Consider a collection of $n$ items, among which exactly $k$ are defective, where $n$ and $k$ are assumed to be large. The goal is to identify all defective items through a testing procedure. Dorfman \cite{dorfman_1943} in 1943 proposed to pool, or group, multiple items together to reduce the overall number of tests that have to be conducted to identify all defective items. He worked under the assumption that a test returns a positive outcome if at least one of its items is defective, and a negative outcome otherwise. This original setting is called \emph{classical or binary group testing}. 
Since the work of Dorfman, a multitude of different versions of group testing have emerged, such as quantitive group testing, where the tests reveal the exact number of defectives participating in it, or noisy group testing, where the tests undergo some noise channel before displaying their result (see \cite{AJS_book} and the references therein).

In 2006, Damaschke \citep{damaschke2006threshold} introduced a further natural generalisation of the classical framework, known as {\em threshold group testing} (TGT). Here, the test outcome depends on a prescribed threshold parameter $t$: a test is positive if at least $t$ items in the pool are defective, and negative otherwise.

In the present article, we study threshold group testing with a constant threshold $t$ and in the sparse regime where $k=n^\theta$ for $\theta \in (0,1)$. We propose a conceptually simple and flexible algorithm that, with high probability as $n \to \infty$, efficiently identifies all defective items while employing a number of tests that we conjecture to be asymptotically minimal.

\subsection{Main Result} \label{SSec:MainResult}

Let $n$ denote the total number of items, which are denoted by $x_1,\dots,x_n$. Each item $x_i$ is associated with a binary label $\SIGMA(x_i)\in\{0,1\}$\footnote{Throughout the paper, we use boldface notation to denote random quantities.}, and we write $\SIGMA=(\SIGMA(x_1),\dots,\SIGMA(x_n))$ for the resulting label vector. 
Items with label $1$ are referred to as \emph{defective}, and we also call $\SIGMA$ the \emph{ground truth}. Additionally we  denote $V_1$ as the set of defective items and $V_0$ as the set of non-defective items.

We focus on a sparse setting where $k$ is known. 
More precisely, we assume that 
$$k=\lfloor n^\theta\rfloor \quad\text{for some fixed }\theta\in(0,1),$$
and that $\SIGMA$ is drawn uniformly at random from all binary vectors of length $n$ with Hamming weight $k$. This can, e.g.,  be achieved by performing a random permutation of the items.

A \emph{pooling scheme} on the items $x_1,\dots,x_n$ consists of a collection of tests $a_1,\dots,a_m$.
Each test $a_j$ corresponds to a multiset $\partial a_j$ of items.
We represent such a scheme by a bipartite multi-graph $\G$, whose first vertex set corresponds to the items and whose second vertex set corresponds to the tests (see Figure~\ref{fig:example}). 

Each pool $a_j$ produces a binary outcome $\hat\SIGMA_j$. Given a fixed threshold $t$, this outcome equals $1$ if at least $t$ of the items contained in $a_j$ are defective (counted by multiplicities), and equals $0$ otherwise. Collecting these outcomes over all pools yields the measurement vector $\hat\SIGMA_\G=(\hat\SIGMA_1,\dots,\hat\SIGMA_m)$.

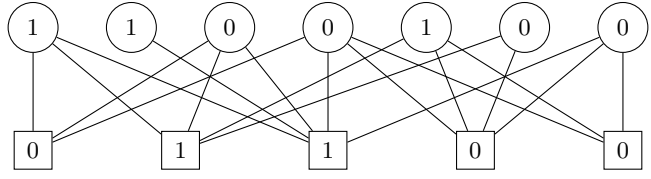
\begin{figure}[ht]
\centering
\begin{tikzpicture}[scale=0.65]
\node[circle, draw, minimum width=0.66cm] (x0) at (0, 0) {$1$};
\node[circle, draw, minimum width=0.66cm] (x1) at (2,0) {$1$};
\node[circle, draw, minimum width=0.66cm] (x2) at (4, 0) {$0$};
\node[circle, draw, minimum width=0.66cm] (x3) at (6, 0) {$0$};
\node[circle, draw, minimum width=0.66cm] (x4) at (8, 0) {$1$}; 
\node[circle, draw, minimum width=0.66cm] (x5) at (10, 0) {$0$};
\node[circle, draw, minimum width=0.66cm] (x6) at (12, 0) {$0$};

\node[rectangle, draw, minimum width=0.5cm, minimum height=0.5cm] (a1) at (0, -2.5) {$0$};
\node[rectangle, draw, minimum width=0.5cm, minimum height=0.5cm] (a2) at (3,-2.5) {$1$};
\node[rectangle, draw, minimum width=0.5cm, minimum height=0.5cm] (a3) at (6, -2.5) {$1$};
\node[rectangle, draw, minimum width=0.5cm, minimum height=0.5cm] (a4) at (9, -2.5) {$0$};
\node[rectangle, draw, minimum width=0.5cm, minimum height=0.5cm] (a5) at (12, -2.5) {$0$};

\path[draw] (x0) -- (a1);
\path[draw] (x0) -- (a2);
\path[draw] (x0) -- (a3);
\path[draw] (x3) -- (a1);
\path[draw] (x1) -- (a3);
\path[draw] (x2) -- (a1);
\path[draw] (x2) -- (a3);
\path[draw] (x2) -- (a2);
\path[draw] (x3) -- (a3);
\path[draw] (x3) -- (a4);
\path[draw] (x3) -- (a5);
\path[draw] (x4) -- (a2);
\path[draw] (x4) -- (a5);
\path[draw] (x4) -- (a4);
\path[draw] (x5) -- (a2);
\path[draw] (x5)-- (a4);
\path[draw] (x6) -- (a3);
\path[draw] (x6) -- (a5);
\path[draw] (x6) -- (a4);
\end{tikzpicture}
\caption{Graphical representation of a pooling scheme for $t=2$: Here, the $n=7$ items are represented by circles, while the $m=5$ pools are represented by rectangles.
The label of an item represents its actual label, while the label of a pool represents the output of its measurement.}
\label{fig:example}
\end{figure}

For non-adaptive binary group testing ($t=1$), the information theoretic lower bound has been established in \cite{aco_2019} as $$m_{\mathrm{inf}}(n,\theta,1)= \max\cbc{\frac{\theta}{(1-\theta)\log^22}, \frac{1}{\log 2}} n^\theta \log \frac{n}{k}.$$ Additionally there are efficient algorithms for binary group testing that use the information-theoretically minimal number of tests.
By adding each item $t$ times to each test it participates in, we see that threshold group testing with a constant threshold $t$ can be reduced to binary group testing. This implies that for any $t$, there are a test design and an efficient algorithm that achieve exact recovery on $(1+\eps)m_{\mathrm{inf}}(n,\theta,1)$ tests. This reduction is, however, only possible when defectives are counted by multiplicities. When this is not the case, TGT may require more tests.

Our main result provides a generalisation of \cite{coja_spiv} for $t>1$ by providing an efficient  algorithm that recovers the ground truth correctly for any $m$ slightly above $\minf(n,\theta, t)$, where $\minf(n,\theta, t)$ is defined as follows: 
For $d>0$, let
\begin{align}\label{eq:cs_inf}
    c_1(d) = \frac{1}{H(\Pr(\Po(d)\leq t-1))} \quad \text{and} \quad c_2(d, \theta) = \frac{\theta}{1-\theta} \bcfr{1}{-d\log(1-\Pr(\Po(d)=t-1))} \, .
\end{align}   
Here $H$ denotes the binary entropy function, while $\Po(d)$ is a Poisson-distributed random variable with mean $d$. 
Using these two functions, we define 
\begin{align}\label{eq:m_inf}
\minf(n,\theta, t)&= \cinf(\theta,t) \, k \log\bc{n/k},
\end{align}
where 
\begin{align} \label{eq:cinf}
    \cinf(\theta, t) = \inf_d \max \{c_1(d), c_2(d,\theta)\}.
\end{align}

\begin{figure}[H]
    \centering
    \includegraphics[width=1\linewidth]{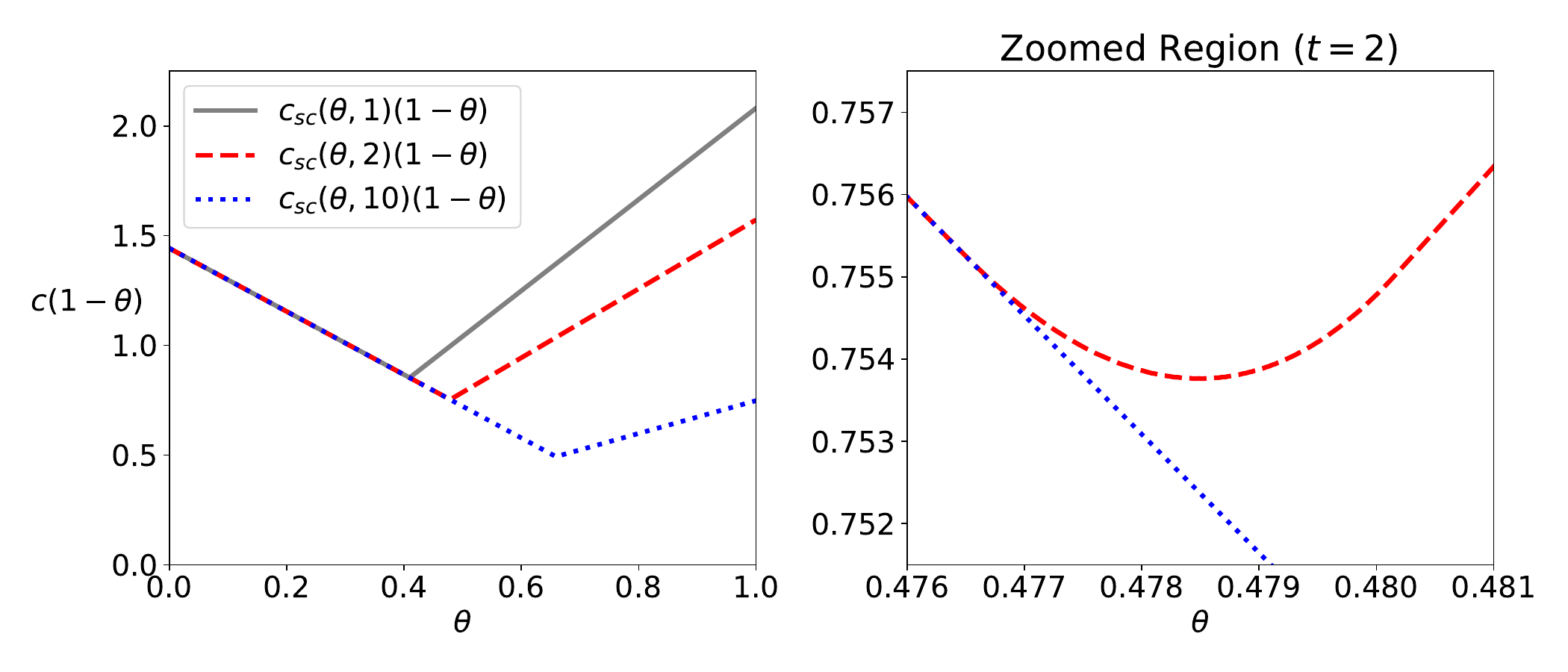}
    \caption{The constant $\cinf$ of visualised for $t=1, t=2, t=10$.}
    \label{fig:thresh}
\end{figure}

\begin{restatable}[]{theorem}{thm_alg_ex}
\label{thm_alg_ex}
	For any $0<\theta<1$, $\eps>0$ and $t\in \NN^+$ fixed, there exists $n_0=n_0(\theta,t,\eps)$ such that, for every $n>n_0$, there exist a randomised test design $\G$ with $m\leq(1+\eps)\minf(n,\theta,t)$ tests and a deterministic polynomial-time inference algorithm $\SPOT$ for which
	\begin{align}\label{eq_alg_ex}
	\pr\brk{\SPOT(\G,\hat{\SIGMA}_\G)=\SIGMA}>1-\eps.
	\end{align}
\end{restatable}

Somewhat surprisingly, \Cref{thm_alg_ex}  and \Cref{fig:thresh} show that there exist values of $\theta$ for which threshold group testing admits exact efficient recovery \emph{with strictly fewer tests than are required in the binary case}:  
The grey line in \Cref{fig:thresh} represents $t=1$. For $t=1$, it is also known  \cite{coja_spiv} that $\minf(n,\theta,1)$ is the information-theoretically optimal number of tests.
For larger values of $t$ and $\theta$, the red (dashed, $t=2$) and blue (dotted, $t=10$) lines fall below the grey line, showing that efficient inference on  a lower number of tests is possible. This complements the results of \cite{vdH2026TGTinfo}, where it was shown that conditionally on an analytic identity, and on the slightly different class of constant-column test designs, TGT undergoes a sharp information-theoretic phase transition for exact recovery at $\minf(n,\theta,t)$.

A further distinction is that the optimal value of $d$ is independent of $\theta$ for $t=1$, and we see a sharp edge between high $\theta$ and low $\theta$ regime. In contrast to this, for $t\geq 2$, the optimal $d(\theta)$ is {\em different} for different $\theta$, leading to a smooth curve instead of a sharp edge, as seen in \Cref{fig:thresh} (right side), even though the curve looks sharp in the figure for the full range $\theta\in[0,1]$.

The \emph{Sp}atially-Coupled \emph{O}utlier \emph{T}esting (\SPOT) algorithm from  \Cref{thm_alg_ex} operates on a test (or factor graph) design called \emph{spatial coupling}, which has successfully been used for efficient recovery in other instances of group testing. 
The design, which we will formally introduce in \Cref{SSec:SpatCoup}, will create a ring-like chain of dependencies. Given this test design, \SPOT\ proceeds in three phases: First, using a relatively large number of tests on a small proportion of the items (the \emph{seed}), it produces an estimate of the seed labels that is already correct \whp\ Next, using the estimates labels of the seed as a kick-start, it generates an estimate of the labels of the remaining items that may contain up to $o(k)$ errors. Finally, it goes through a \emph{cleaning phase} during which all previous misspecifications are removed. 

At a first look,  \SPOT\ 
proceeds in a similar vein to the information-theoretically optimal algorithm \SPIV\ for classical binary group testing, as it is based on a spatially coupled test design, an approximate recovery and a cleaning phase. However, previous spatially coupled algorithms for group testing, such as \SPIV\ and \SPEX\ for binary and \SCIENT\ for quantitative group testing, rely on forming {\em weighted sums} of test outcomes as a score function in the second stage to obtain an approximate solution \cite{coja_spiv, coja2025noisy, hahnklimroth_QGT}. This procedure simulates one step of a message passing algorithm called {\em belief propagation}, where the messages from the items to the tests consist of weighted sums. This algorithm is generally only known to work on problems where the underlying structure is a tree, so that additional calculations are needed to justify its use, as well as a calculation of the optimal weights discussed above.

The main novelty of our approach is that it avoids such weighted sums altogether. Instead, we analyse the distributions of individual test outcomes, conditioned on the test outcome and the number of items already classified in preceding compartments. 
This perspective significantly simplifies the analysis and enables a transferral of the overall strategy to more complex settings.

\paragraph{\bf Applications of Threshold Group Testing.}
Threshold group testing provides a modelling framework for problems where individual contributions are weak and only detectable once a threshold is reached. In healthcare, for example, conditions such as a stroke are assessed using a decision threshold to determine when to take action \cite{gage2001validation,malioutov2017learning}. TGT also captures a class of graph discovery problems in which the goal is to uncover a hidden graph, as shown by Chen and Fu \cite{chen2009nonadaptive}. Each test corresponds to selecting a subset of edges, and the only feedback is whether it contains at least $t$ edges from the hidden graph.
In communication systems, such pooling designs have been used to detect jamming attacks in multi-radio wireless sensor networks \cite{shin2009reactive}, as well as in multiple-packet reception settings, where successful detection depends on the total signal strength rather than individual transmission \cite{chan2012carrier,georgiadis1982collision,ghez1989optimal}.

\subsection{Previous work}
This subsection provides a short overview of the literature that is closest to our regime and of algorithmic tools that we use.

\paragraph{Deterministic recovery.} 
The worst-case version of the threshold group testing problem, where the goal is to correctly identify the ground truth with probability exactly equal to one, has been studied since the work of Damaschke \cite{damaschke2006threshold}. 
For this problem, De Marco et al. \cite{demarco2020subquadratic} proved an upper bound on the number of pools of $O((k+t^2)\log n)$ in the non-adaptive setting, where $t$ is a constant threshold and $k$ is in the sparse regime\footnote{Throughout this work, asymptotic notation is interpreted in the limit as $n \rightarrow \infty$. Specifically, $o(1)$ denotes any term tending to zero, while $\omega(1)$ denotes any term diverging to $\infty$. We write $f(n) = O(g(n))$ to denote that there exists a constant $b > 0$ and an integer $n_0 \in \mathbb{N}$ such that $f(n) \leq b \cdot g(n)$ for all $n \geq n_0$. Similarly, $f(n) = \Omega(g(n))$ indicates that there exists $b > 0$ and $n_0 \in \mathbb{N}$ such that $f(n) \geq b \cdot g(n)$ for all $n \geq n_0$. Finally, $f(n) = \Theta(g(n))$ indicates that $f(n) = O(g(n))$ and $f(n) = \Omega(g(n))$ simultaneously hold.}.
Most recently, Bui et al.~\cite{bui2024efficient} provided the first efficient algorithms for non-adaptive and adaptive threshold group testing with general thresholds. Their results establish explicit upper bounds on the number of tests required for deterministic exact recovery in the worst-case setting, with polynomial-time decoding. In particular, for the non-adaptive case, they obtain bounds of order $O\!\left(k (k-t)^2 \log(n/k)\right)$ tests for general thresholds, with improved guarantees in special cases such as $t=2$ or when allowing multiple stages. While these results provide upper bounds on the number of pools required to infer the ground truth, they focus on the zero-error setting rather than recovery \whp 

\paragraph{Recovery \whp} Recovery \whp\ refers to a setting where the probability of errors vanishes as $n\rightarrow \infty$. When analysing recovery \whp, we assume that the number of defective items is known and that the true label vector is chosen uniformly among all vectors of this Hamming weight. 
As $n$ tends to infinity, the probability of exact recovery given the set number of tests should tend to one.
Key results include the work of Chan et al. and Reisizadeh et al.\ \cite{chan2013stochastic, reisizadeh2018sub}. Chan et al.\ propose a two-round procedure using $O(k \log n)$ pools, but with large constants involving an error term denoted by $\epsilon$ in terms of which the number of pools needed grows like $16e^2 k \log n + O\left( \log (1/\epsilon) k  \right)$. Reisizadeh et al.\ provide a single-round, non-adaptive algorithm using $O(k \sqrt{t} \log^3 n)$ pools, but it requires a large look-up matrix containing $O(t \log n) \times \tbinom{n}{t}$ elements.

Recent work by van der Hofstad et al. \cite{vdH2026TGTinfo} has established that for $t=2$, and $t\geq 3$ under an additional analytic assumption, there exists a sharp information-theoretic threshold for threshold group testing, at $m_{\mathrm{inf}}$
for the constant column test design, where each item joins exactly $\Delta$ tests that are chosen uniformly at random. 
The information-theoretic result \cite{vdH2026TGTinfo} yields a pooling scheme on $(1+\eps) \minf(n,\theta, t)$ such that \whp{} there is exactly one defective set that matches the test results. Therefore, a naive algorithm to find the ground truth would be to go through all $\tbinom{n}{k}$ such vectors and see which one explains the measurements, thus producing an algorithm that has time-complexity exponential in $k \log n$. 

To obtain a polynomial-time algorithm, different approaches have been proposed, alas at the expense of a greater number of measurements. As noted, the algorithms by Chan et al.\ and Reisizadeh et al.\ \cite{chan2013stochastic, reisizadeh2018sub} rely either on multiple stages, large constants, or poly-logarithmic factors exceeding the sharp threshold $\minf(n,\theta, t)$.

In this paper, we present the polynomial-time algorithm \SPOT\ that \whp\ recovers the ground truth with the optimal number of tests $\minf(n,\theta,t)$ for the constant column design. 
However, since the lower bound has been proven only for the constant column design,
it remains an open question whether our algorithm indeed is optimal with respect to arbitrary test designs.

Shortly after the appearance of the present article, Tran, McMorrow and Scarlett \cite{selectable_2026} introduced a more general variant of TGT, termed group testing with selectable thresholds. In their model, each test $a$ can be assigned its own threshold $t_a$, while TGT is recovered by choosing $t_a=t$ for all tests. As in the present work, \cite{selectable_2026} studies recovery \whp\ using non-adaptive test designs. 
The authors consider both bounded and unbounded maximum thresholds, and propose several algorithms that first identify an almost complete set of defectives using standard group testing queries and then recover the remaining defectives through additional tests with different thresholds. When restricted to the case $t_a=t$ for all tests, their converse Theorem 5 establishes an information-theoretic lower bound for a broader class of test designs than the one of \cite{vdH2026TGTinfo}. For sufficiently large $\theta$, this bound matches the lower bound of \cite{vdH2026TGTinfo}. While the results of \cite{selectable_2026} do not fully  resolve the question whether the lower bound from \cite{vdH2026TGTinfo} is tight for arbitrary test designs, they provide further evidence towards this claim.


\subsection{Spatial Coupling} \label{SSec:SpatCoup}

Spatial Coupling is a technique used in coding theory, particularly LDPC codes \cite{Kudekar_2010_2,Kudekar_2011,kudekar_2013} and has previously been successfully used in group testing \cite{coja_spiv,coja2025noisy}. 
Recall that we have $n$ items $x_1, \dots, x_n$ and let us denote the set of these items as $V = \{x_1, \dots, x_n\}$. 
Our aim is to build a pooling scheme which provides a spatial dependence between items and pools that aids our algorithm. 

\subsubsection{Parameters}
Let $d^*$ be the unique minimiser of \eqref{eq:cinf} (see \cite[Lemma 2.4]{vdH2026TGTinfo}) and $d'$ be such that $\pr\bc{\Po(d')\leq t-1} = 1/2$. For positive constants $c, c' > 0$ we choose
\begin{align}\label{Eq:Param}
m = c k \log(n/k), \quad \quad \quad \quad \quad \quad \Delta = c d^* \log(n/k),
\end{align}
as well as
\begin{align}\label{Eq:Param_seed}
\mseed = c' k \log n, \quad \quad \quad \quad \quad \quad \Delta' = c' d' \log n.
\end{align}
The parameters $c,c'$ control the numbers of pools in the spatially coupled test design and a so-called \emph{seed} segment, respectively, while $d^*, d'$ calibrate the item degrees accordingly. 

\subsubsection{Construction of the Pooling Scheme}
The following description of the spatially coupled test design follows the description of \cite{coja_spiv}. To begin, we partition the items in $V$ into 
\begin{align}\label{eq:ell}
    \ell = \ceil{\sqrt{\log n }}
\end{align} 
\emph{compartments} $V[1], \ldots, V[\ell] \subset V$ of (almost) equal size $|V[i]| \in \{\lfloor n/\ell \rfloor, \lceil n/\ell \rceil\}$. The partitioning of items into the compartments induces a partitioning of defective items as well, which we will denote $V_1[1], \ldots, V_1[\ell]$.
Similarly, for a given number of tests $m$ divisible by $\ell$, we take $m$ pools and partition them into $\ell$ compartments such that each of the compartments contains exactly $m/\ell$ pools. 
The $\ell$ pool compartments are denoted by $F[1], \ldots, F[\ell]$. 
The compartments introduce a natural order to stacks of items and tests, and moreover are arranged in an overlying ring structure, so that compartments $V[\ell]$ and $F[\ell]$ are followed by $V[1]$ and $F[1]$. 

Further, let \begin{align*}
    s = \lceil \log\log n \rceil
\end{align*}
 and $\Delta = \Theta(\log n)$ be as in \eqref{Eq:Param} such that $\Delta$ is divisible by $s$. 
Then each item $x\in V[i]$ independently joins $\Delta/s$ uniformly chosen pools from each of the test compartments $F[i-j-1]$ for $1\leq j \leq s$, where pools are chosen with replacement. Therefore, $s$ is also called the {\em sliding window}. 

Finally, we add an additional test compartment $F[0] = \{a'_1,\ldots, a'_{\mseed}\}$ of size $\mseed=o(m)$, whose tests are only connected to the first $s$ compartments. 
Each item in $V[1], \dots, V[s]$ then independently joins $\Delta'= \Theta(\log n)$ tests of $F[0]$ that are chosen uniformly at random and with replacement.

In the following, we denote the \emph{random bipartite graph} arising from this construction by $\vec G_{\mathrm{sc}}$.  
Locally, $\vec G_{\mathrm{sc}}$ looks almost like a random constant column graph, while globally, it possesses a ring structure with interlocking edges between compartments, as visualised in \Cref{Fig_spatial_coupling_idea}. Note that $\vec G_{\mathrm{sc}}$ may contain multi-edges as items sample tests with replacement.

Turning to inference algorithms on $\vec G_{\mathrm{sc}}$, the first $s$ compartments will form the seed $V_{\mathrm{seed}}=V[1] \cup \cdots \cup V[s]$, and be inferred using the tests from $F[0]$ only. 
After obtaining the seed labels, the partitioning of items and pools will allow us to successively infer the labels of compartments $V[i]$ for $i \in \{s+1, \ldots, \ell\}$, using the information of previous compartments along the way.

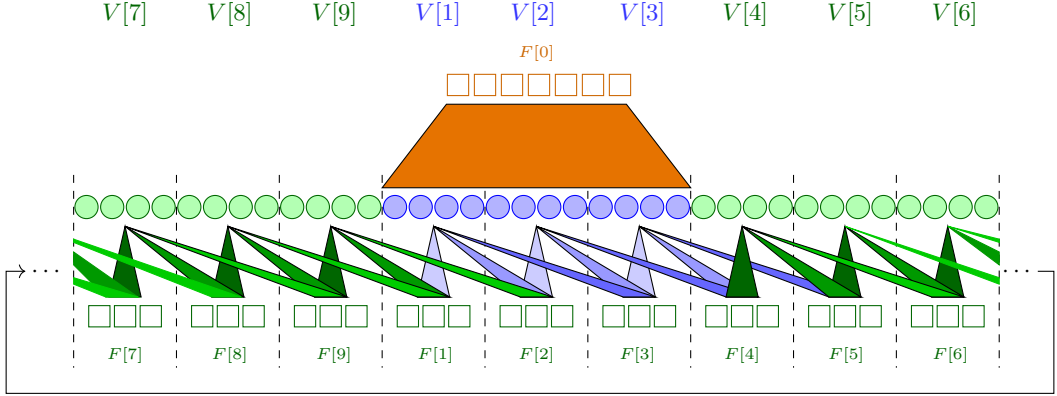
\begin{figure}[ht!]
\begin{tikzpicture}[scale=0.85]

\foreach \i in {1,...,12}
{
        \def\lab{x_\i};
        \node[circle,draw=black!60!green,fill=green!30,minimum size=1] (\lab) at (0.4*\i,0) {};
}
\foreach \i in {13,...,24}
{
        \def\lab{x_\i};
        \node[circle,draw=black, color=blue,minimum size=1, fill=blue!30] (\lab) at (0.4*\i,0) {};
}
\foreach \i in {25,...,36}
{
        \def\lab{x_\i};
        \node[circle,draw=black!60!green, fill=green!30,minimum size=1] (\lab) at (0.4*\i,0) {};
}
\foreach \i in {0,...,9}
{
        \def\x{4*\i};
        \draw[dashed] (0.4*\x+0.2,0.5) -- (0.4*\x+0.2,-2.5);
}
\foreach \i in {1,...,9}
{
        \def\labx{c_\i};
        \def\labaone{done_\i};
        \def\labatwo{dtwo_\i};
        \pgfmathsetmacro{\xcoord}{0.4*(\i+1.5)+3*0.4*(\i-1))};

        \coordinate (\labx) at (\xcoord,-0.3);
        \coordinate (\labaone) at (\xcoord-0.25,-1.4);
        \coordinate (\labatwo) at (\xcoord+0.25,-1.4);
}

\foreach \j in {2,3,4,5,6,7,8,9,10}{
    \pgfmathsetmacro{\ione}{2*\j-1};
    \pgfmathsetmacro{\itwo}{2*\j};
    \def\laba{a_\ione};
    \def\labb{a_\itwo};
    \pgfmathsetmacro{\xcoordone}{0.4*4*(\j-2)+1.4};
    \pgfmathsetmacro{\xcoordtwo}{0.4*4*(\j-2)+1.0};
    \pgfmathsetmacro{\xcoordthree}{0.4*4*(\j-2)+0.6};
    \node[rectangle, minimum size=8,draw=black!60!green] (\laba) at (\xcoordone,-1.7){};
    \node[rectangle, minimum size=8,draw=black!60!green] (\labb) at (\xcoordtwo,-1.7){};
    \node[rectangle, minimum size=8,draw=black!60!green] (\laba) at (\xcoordthree,-1.7){};

}

\foreach \i in {4,5,6}
{
    \filldraw[fill=blue!20] (c_\i) -- (done_\i) -- (dtwo_\i) -- cycle;
}

\foreach \i in {4,5,6}
{
    \pgfmathsetmacro{\x}{\i+1};
    \filldraw[fill=blue!40] (c_\i) -- (done_\x) -- (dtwo_\x) -- cycle;
}
\foreach \i in {4,5,6}
{
    \pgfmathsetmacro{\x}{\i+2};
    \filldraw[fill=blue!60] (c_\i) -- (done_\x) -- (dtwo_\x) -- cycle;
}

\foreach \i in {1,2,3,7,8,9}
{
    \filldraw[fill=black!60!green] (c_\i) -- (done_\i) -- (dtwo_\i) -- cycle;
}

\foreach \i in {1,2,3,7,8}
{
    \pgfmathsetmacro{\x}{\i+1};
    \filldraw[fill=black!40!green] (c_\i) -- (done_\x) -- (dtwo_\x) -- cycle;
}
\foreach \i in {1,2,3,7}
{
    \pgfmathsetmacro{\x}{\i+2};
    \filldraw[fill=black!20!green] (c_\i) -- (done_\x) -- (dtwo_\x) -- cycle;
}

\filldraw[black!40!green] (0.2,-0.7) -- (1.2,-1.4)--(0.7,-1.4)--(0.2,-0.9)--cycle;
\filldraw[black!20!green] (0.2,-1.1) -- (1.2,-1.4)--(0.7,-1.4)--(0.2,-1.2)--cycle;
\filldraw[black!20!green] (0.2,-0.5) -- (2.8,-1.4)--(2.35,-1.4)--(0.2,-0.55)--cycle;
\filldraw[black!40!green] (13.8,-0.3) -- (14.6,-0.7)--(14.6,-0.9)--cycle;
\filldraw[black!20!green] (13.8,-0.3) -- (14.6,-0.5)--(14.6,-0.55)--cycle;
\filldraw[black!20!green] (12.2,-0.3) -- (14.6,-1.1)--(14.6,-1.2)--cycle;



\filldraw[black!10!orange, draw=black] (9.8,0.3)--(7.4,0.3) -- (5.0,0.3) -- (6.0, 1.6)--(8.8,1.6)--cycle;

\node (A)[text=black!60!green] at (1,3) {$V[7]$};
\node (B)[text=black!60!green] at (2.63,3) {$V[8]$};
\node (C)[text=black!60!green] at (4.26,3) {$V[9]$};
\node[text=blue!80] (D) at (5.87,3) {$V[1]$};
\node[text=blue!80] (E) at (7.35,3) {$V[2]$};
\node[text=orange!80] (Z) at (7.4,2.4) {{\tiny\textcolor{black!20!orange}{$F[0]$}}};
\node[text=blue!80]  (F) at (9.05,3) {$V[3]$};
\node (G)[text=black!60!green] at (10.65,3) {$V[4]$};
\node (H)[text=black!60!green] at (12.28,3) {$V[5]$};
\node (I)[text=black!60!green] at (13.91,3) {$V[6]$};

\node (J)[] at (1,-2.3) {{\tiny\textcolor{black!60!green}{$F[7]$}}};
\node (K)[] at (2.63,-2.3) {{\tiny\textcolor{black!60!green}{$F[8]$}}};
\node (L)[] at (4.25,-2.3) {{\tiny\textcolor{black!60!green}{$F[9]$}}};
\node (M)[] at (5.83,-2.3) {{\tiny\textcolor{black!60!green}{$F[1]$}}};
\node (N)[] at (7.4,-2.3) {{\tiny\textcolor{black!60!green}{$F[2]$}}};
\node (O)[] at (8.97,-2.3) {{\tiny\textcolor{black!60!green}{$F[3]$}}};
\node (P)[] at (10.6,-2.3) {{\tiny\textcolor{black!60!green}{$F[4]$}}};
\node (Q)[] at (12.22,-2.3) {{\tiny\textcolor{black!60!green}{$F[5]$}}};
\node (R)[] at (13.84,-2.3) {{\tiny\textcolor{black!60!green}{$F[6]$}}};
\node at (-0.2,-1){$\cdots$};
\node at (14.9,-1){$\cdots$};

\foreach \i in {1,...,7}
{
    \node[rectangle, minimum size=8,draw=black!20!orange] at (5.76+0.42*\i,1.9){};
}

\draw[->] (15.2, -1) -- (15.45, -1) -- (15.45, -2.9) -- (-0.85, -2.9) -- (-0.85, -1) -- (-0.55, -1);
\end{tikzpicture}
\caption{Schematic representation of the pooling scheme with $n = 36$ items, $\ell = 9$ compartments, $m=27$ measurements and a sliding window of size $s = 3$. 
The number of (blue) seed items, whose label needs to be determined first, is $12$ in this example. 
The additional orange pools on top of the graph represent the $\mseed$ tests required to identify the seed prior.
}
\label{Fig_spatial_coupling_idea}
\end{figure}

\subsubsection{Degrees and Numbers of Defectives per Compartment}
The spatially coupled test design is different from the classical constant column design, where every item chooses its $\Delta$ neighbours uniformly among all tests. However, the local structure of $\vec G_{\mathrm{sc}}$ resembles the constant columns design, so that many of its properties carry over.

For $x_i \in V$, let $\partial x_i$ be the multiset of tests $a_j$ that include $x_i$. Thus, $\abs{\partial x_i} = \Delta' + \Delta$ for $x_i \in V_{\mathrm{seed}}$ and $\abs{\partial x_i} = \Delta$ else. 
Similarly, for each test $a_j, j \in [m]$, $\partial a_j$ denotes the multiset of items included in $a_j$, and we set $\vGamma_j = \abs{\partial a_j}$. For $a_j' \in F[0]$, $\partial a_j'$ denotes the multiset of items included in $a_j'$, and we set $\vGamma_j' = |\partial a_j'|$.
Then, for $j \in [m]$,
$$\vGamma_j \sim \sum_{i=1}^s \Bin(n\Delta/\ell s, \ell/m) =: \sum_{i=1}^s \vGamma_j^{(i)} \qquad \text{and} \qquad \sum_{j=1}^m \vGamma_j = n \Delta.$$  
Analogously, for $j \in [\mseed]$,
$$\vGamma_j' \sim  \Bin\bc{|V_{\mathrm{seed}}|\Delta', \frac{1}{\mseed}} \qquad \text{and} \qquad \sum_{j=1}^{\mseed} \vGamma_j' = \abs{V_{\mathrm{seed}}} \Delta'.$$
Finally, we denote the vector of test degrees for the seed tests by $\vGamma_{\mathrm{seed}} = \bc{\vGamma_j'}_{j \in [\mseed]},$ and the vector of test degrees for the bulk tests by $\vGamma = \bc{\vGamma_j}_{j \in [m]}$.

We define $\vGamma_{\min} = \min_{j \in [m]} \vGamma_j$ and $\vGamma_{\max} = \max_{j \in [m]} \vGamma_j$ to capture the range of test degrees. 
Finally, the analysis at times will rely on contributions per individual compartment, thus define $$\vGamma_{\min}^{\mathrm{comp}} = \min_{j \in [m], i \in [s]} \vGamma_j^{(i)} \qquad \text{and} \qquad \vGamma_{\max}^{\mathrm{comp}} = \max_{j \in [m], i \in [s]} \vGamma_j^{(i)}.$$ 
The Chernoff bound for the Binomial distribution (see \Cref{lem_chernoff_2}) then immediately implies the following.

\begin{lemma}[{\cite[Proposition 4.1]{coja_spiv}},{\cite[Lemma II.4]{aco_2019}}] \label{Lemma_GammaMinMax}
    Assume $m = ck \log(n/k)$ and $\Delta = c d \log(n/k)$ for constant $c,d>0$. With probability at least $1-o(n^{-2})$, 
    \begin{enumerate}
        \item the defective item count per compartment satisfies 
        \begin{equation} \label{eqG1}
            \frac{k}{\ell} - \sqrt{\frac{k}{\ell}} \log n \leq \min_{i \in [\ell]} \abs{V_1[i]}  \leq \max_{i \in [\ell]} \abs{V_1[i]} \leq \frac{k}{\ell} + \sqrt{\frac{k}{\ell}} \log n;
        \end{equation}
        \item the test degrees in a single compartment satisfy
        \begin{align}\label{eq_comp_minmax_G}
        \frac{\Delta n}{m s}-\sqrt{\frac{\Delta n}{ms}}\log n \leq \vGamma_{\min}^{\mathrm{comp}} \leq \vGamma_{\max}^{\mathrm{comp}} \leq \frac{\Delta n}{ms}+\sqrt{\frac{\Delta n}{ms}}\log n;
        \end{align}
        \item the overall test degrees satisfy
        \begin{align} \label{eq_minmax_G}
        \frac{\Delta n}{m}-\sqrt{\frac{\Delta n}{m}}\log n \leq \vGamma_{\min} \leq \vGamma_{\max} \leq \frac{\Delta n}{m}+\sqrt{\frac{\Delta n}{m}}\log n. 
        \end{align}
    \end{enumerate}
\end{lemma}

\begin{remark}\label{rem_seed_deg}
    Analogous concentration bounds as in \Cref{Lemma_GammaMinMax}, item 3., hold for $\vGamma'_{\min} = \min_{j \in [\mseed]} \vGamma'_j$ and $\vGamma'_{\max} = \max_{j \in [\mseed]} \vGamma'_j$. 
\end{remark}

In the following, we let $\cE$ denote the event that all points of \Cref{Lemma_GammaMinMax} are satisfied.

\subsection{Diagnosing the Seed (Basic Thresholding)} 
\newcommand{\bth}{{\rm{BTH}}}

The first step of \SPOT\ diagnoses every item of $V_{\mathrm{seed}}$ using only the $\mseed = o(\minf)$ tests in $F[0]$ by performing a basic thresholding (BTH) algorithm on every item. 
The score function in \bth\ counts the number of positive tests each item $x$ partakes in, and is defined as 
\[ 
\cT_x = \sum_{a \in \partial x \cap F[0] } \hat{\SIGMA}_a.
\] 
Thus, \bth\ can be seen as a generalised version of a simple algorithm for binary group testing ($t=1$), \COMP\ \cite{AJS_book}, that diagnoses any item that is included in at least one negative test as non-defective and all other items as defective. 
Since every seed item participates in $\Delta'$ tests from $F[0]$, we choose the following (parametrised) threshold for classification of an item as defective: 
 $$T(d) = \Delta' \bc{1- P_{\sss \leq t-1}(d) + \alpha^{\mathrm{*}}P_{t-1}(d) }, \label{Eq:Thresh}$$
 where
\begin{align} \label{def_thr_para}
    \alpha^{\mathrm{*}} &= \frac{1}{1+\sqrt{\theta}}, \qquad
   P_{t-1}(d)=\pr\bc{\Po(d)=t-1} 
   \qquad \mbox{and} \qquad P_{\sss \leq t-1}(d)=\pr\bc{\Po(d)\leq t-1}.
\end{align}
Further, let $d'$ be defined by 
\begin{align}\label{def_dprime}
  d'=\argmin_d \abs{P_{\sss \leq t-1}(d) - \frac 1 2 }.  
\end{align}
Then, based on the threshold $T(d')$, \bth\ will classify as follows: 

\begin{algorithm}[H]
    \KwData{$n, G=(V,F), \hat{\SIGMA}$}
        \For{$x \in V$}{
            $\tau_x = \vecone  \cbc{\cT_x>T(d')}$ 
        }   
    \Return $\tau_x$.\\
\caption{Basic thresholding \bth, used in \SPOT, phase 1.}
\label{Algo_Naive}
\end{algorithm}

The algorithm exploits the expected differences between a test including a defective as opposed to a non-defective item. 
Taking into account that the exact number of defective items in the seed is random and unknown, the following proposition shows that \bth\ indeed returns the correct states for all items in the seed \whp:
\begin{restatable}[]{proposition}{thmnaivethresholding}
    \label{prop_naive_thresholding}
   Let $\theta \in (0,1)$. For $N \in \NN$, let $\sigma \in \{0,1\}^{N}$ be a vector of unknown Hamming weight $K' \in \brk{K-\sqrt{K}\log N, K + \sqrt{K}\log N}$, where $K = N^{\theta}\bc{\frac{s}{\ell}}^{1-\theta}$. For $\delta >0$, let $M = (1+\delta)\mseed$ with 
    \begin{align}\label{eq:mseed}
           \mseed = \cseed  K \log \frac{N}{K},  \qquad\qquad \cseed = \frac{2(1+\sqrt{\theta})}{1-\sqrt{\theta}} \cdot \frac{1}{d' P_{t-1}(d')^2}. 
   \end{align} 
   Finally, let $\vec G$ be the constant column design on items $\{x_1, \ldots, x_N\}$ and $M$ pools, where every item chooses  $\Delta' = c_{\mathrm{seed}} d' \log \frac{N}{K}$ of the tests uniformly at random without replacement. Then \whp{} as $N \to \infty$, the algorithm $\bth$ (\Cref{Algo_Naive}) on input $N, \vec G$ and $\hat \sigma_{\vec G}$ recovers $\sigma$.
\end{restatable}

We will later apply \Cref{prop_naive_thresholding} to $\vec G = (V[0], F[0])$. Thus, $N = sn/\ell$ and \whp, the number of defective items in the seed is concentrated around $K=(sk/\ell)^{\theta} = N^{\theta} (s/\ell)^{1-\theta}$. \Cref{prop_naive_thresholding} then ensures that $\bth$ diagnoses all items in $V[0]$ correctly \whp\ using only $o(\minf)$ tests, since $\abs{V_1[0]} = o(k)$. 
The proof of \Cref{prop_naive_thresholding} is provided in \Cref{SSec:NaiveThersholding}.

\subsection{Approximate Recovery} 

After the first phase, \whp\ the correct labels of all items in the seed are known. The second phase of \SPOT\ diagnoses the remaining items compartment by compartment, using information of previously diagnosed compartments.
Suppose we wish to determine the status of an item $x$ in compartment $i>s$, then we already have an approximation of the status of items in compartment 1 to $i-1$. 
Now consider the different types of tests we can observe for item $x$. 
While the tests are in different compartments and can be positive or negative, the number of defective items from previous compartments can also differ.

For any $i\in\{1,\dots,l\}$, let $F_{\geq t}[i]$ be the tests in compartment $i$ that have at least $t$ edges connecting them to a defective item (i.e., positive tests), while $F_{< t}[i]$ are the tests with less that $t$ edges connecting to defective items (i.e., negative tests). Given an estimate $\tau$, for all $j \in \{1,\dots,s\}, \; r \in \{0,\dots,t-1\}$ and every item $x\in V[i]$, let $\vW^{r, +}_{x,j}(\tau)$ and $\vW^{r, -}_{x,j}(\tau)$ be the numbers of positive/negative neighbours of $x$ in compartment $F[i+j-1]$ that contain exactly $r$ previously identified as defective items according to $\tau$, i.e.,

	\begin{align} \label{eq_w}
    \begin{split}
    		\vW^{r, +}_{x,j}(\tau) &= \abs{ \cbc{ a \in F_{\geq t}[i+j-1] \cap \partial x \; : \abs{ \cbc{ v\in \partial a \cap \bigcup_{h=1}^{s-1} V[i-h] \; :  \tau_v=1 } } = r \;   } }; \text{\,and} \\
		\vW^{r, -}_{x,j}(\tau) &= \abs{ \cbc{ a \in F_{< t}[i+j-1] \cap \partial x  \; : \, \abs{ \cbc{ v\in \partial  a \cap \bigcup_{h=1}^{s-1} V[i-h] \; :  \tau_v=1 } } = r \;     } }. 
     \end{split}
	\end{align}
In \eqref{eq_w}, we count multi-edges separately, i.e., each copy of an edge to a defective item contributes 1 in the sum over $\partial  a$. 

Since the vast majority of items is non-defective, the main classification task is to avoid misclassifying non-defective items. 
In order to do so, we will benchmark all $\vW^{r, +}_{x,j}(\tau)$ and $\vW^{r, -}_{x,j}(\tau)$-counts against the asymptotic expected number $ W_{1,j}^{r,\pm}$  of the corresponding counts for defective items.  
Our algorithm will then \emph{only} tentatively classify $x$ as defective if, for each type of test, the number of tests of that type in which $x$ participates matches the expected number for a defective item up to a certain deviation that we will tune appropriately. 
That is, \SPOT\ will only tentatively classify an item as defective if, given all the knowledge at the time of classification, $x$ looks like a defective item.

We next describe how we choose which items we initially classify as defective given $\vW^{r, +}_{x,j}(\tau)$ and $\vW^{r, -}_{x,j}(\tau)$. For this, and for each $r \in \{0,\dots, t-1\}, j \in \cbc{1,\dots,s}$, let
    	\begin{align}
		W_{0,j}^{r,+} = q_{0,j}^{r,+} \frac{\Delta}{s}, \qquad
		W_{0,j}^{r,-} = q_{0,j}^{r,-} \frac{\Delta}{s}, \qquad
        W_{1,j}^{r,+} = q_{1,j}^{r,+} \frac{\Delta}{s}, \qquad
		W_{1,j}^{r,-} = q_{1,j}^{r,-} \frac{\Delta}{s};
	\end{align}
    where
		\begin{align} 
			q_{0,j}^{r,+} &= \Pr\bc{\Po\bc{\frac{d^*(s-j)}{s}} = r} \Pr\bc{\Po\bc{\frac{d^*j}{s}} \geq t-r},   \label{eq_def_q0p}\\
			q_{0,j}^{r,-} &= \Pr\bc{\Po\bc{\frac{d^*(s-j)}{s}} = r} \Pr\bc{\Po\bc{\frac{d^*j}{s}} < t-r}, \label{eq_def_q0m} \\
			q_{1,j}^{r,+} &=  \Pr\bc{\Po\bc{\frac{d^*(s-j)}{s}} = r} \Pr\bc{\Po\bc{\frac{d^*j}{s}} \geq t-r-1}, \text{\, and} \label{eq_def_q1p} \\
			q_{1,j}^{r,-} &=  \Pr\bc{\Po\bc{\frac{d^*(s-j)}{s}} = r} \Pr\bc{\Po\bc{\frac{d^*j}{s}} < t-r-1}. \label{eq_def_q1m}
		\end{align}
        
 The idea is that the numbers $W^{r, \pm}_{0/1,j}(\tau)$ approximate the asymptotic expected values of $\vW^{r, \pm}_{x,j}(\tau)$ if $\SIGMA_x = 0$ or $\SIGMA_x = 1$, respectively, which can be understood as follows: Roughly, the quantities $q_{0/1,j}^{r,+/-}$ approximate the probabilities that a given neighbour of $x \in V[i]$ in compartment $F[i+j-1]$ is $+$ or $-$ (i.e., positive or negative), and has exactly $r$ defective neighbours in compartments $V[i-(s-j)],\ldots, V[i-1]$, if $x$ is $0/1$. 
 
 As an example, consider $q_{0,j}^{r,+}$, which approximates the probability that a given neighbour $a \in F[i+j-1]$ of $x$ has a positive outcome and has exactly $r$ defective neighbours in compartments $V[i-(s-j)],\ldots, V[i-1]$, if $x$ is $0$. In this case, the first factor in \eqref{eq_def_q0p} approximates the probability that $a$ has exactly $r$ defective neighbours in compartments $V[i-(s-j)],\ldots, V[i-1]$, while the second factor approximates the probability that $a$ has at least $t-r$ defective neighbours in compartments $V[i],\ldots,V[i+j-1]$, so that the test result would be 1. The Poisson distributions arise from the observation that, conditionally on the exact numbers of defective items in compartments $V[i-(s-j)],\ldots, V[i+j-1]$, as well as the numbers of neighbours of $a$ in these compartments, the numbers of defective neighbours of $a$ in compartments $V[i-(s-j)],\ldots, V[i+j-1]$ follow independent Binomial distributions with many trials and a small success probability, which suggests that a Poisson approximation should hold. The parameters can be found by observing that $a$ has approximately $\Delta n/(ms) = dn/(ks)$ neighbours in a single compartment, each of which is defective with probability approximately $k/n$. While not constituting an exact proof, this motivates the above expressions. An important ingredient in our proof is to make this intuition precise.

Finally, let
\begin{align}
    \zeta = \frac{1}{\log \log \log n}.
\end{align}
We now formalise the decision rule in the initial phase of our algorithm through the following steps of the algorithm:

\IncMargin{1em}
	\begin{algorithm}[H]
		 \KwData{$n$, $\vec G_{\mathrm{sc}}$, $\hat{\SIGMA}$}
        \setcounter{AlgoLine}{1}
            $\tau' = \bth( \G=(V_{\mathrm{seed}}, F[0]), {\hat{\vec\sigma}}_{F[0]})$\\
            $\tau_x = \tau'_x$ for all $x \in V_{\mathrm{seed}}$\\
		\For{$i=s+1, \dots, \ell $}{
			\For{$x \in V[i]$}{
				$\tau_x = 0$ \tcp*[h]{classify as non-defective} \\
				\If{For all $j\in [s]$, $r\in \{0,\dots, t-1\}$:  $\vW_{x,j}^{r,+}(\tau) \geq (1 - \zeta) W_{1,j}^{r,+} \mbox{ and }\vW_{x,j}^{r,-}(\tau) \leq (1 + \zeta) W_{1,j}^{r,-}$ }{
							$\tau_x = 1$ \tcp*[h]{classify as defective}}
		}}
		\Return $\tau$
		\caption{\SPOT, phase 2.}
	\end{algorithm}
\DecMargin{1em}

After obtaining the correct labels for the seed, the second phase of \SPOT\ diagnoses all other items. 
The following proposition shows that the second phase of \SPOT\ returns an approximate solution with at most $o(k)$ errors \whp{}, so that it classifies {\em most} items correctly:

\begin{restatable}[]{proposition}{propapprox} \label{prop_step1}
	Suppose that $m \geq (1+\eps)\minf$.  
Then \whp\ the output $\tau$ of phase 2 of \SPOT\ satisfies
	$$\sum_{x\in V}\vecone\cbc{\tau_x\neq\SIGMA_x}\leq k\exp\bc{-\Omega\bcfr{\log n}{(\log\log n)^{2t+3}}}.$$
\end{restatable}

The proof of \Cref{prop_step1}, given in \Cref{Sec:Approx}, proceeds in two steps. 
First, we slightly overestimate the number of misclassified defective and non-defective items under the assumption that the algorithm has exact knowledge of previous compartments $\vW_{x,j}^{r,\pm}(\SIGMA)$, rather than the approximations $\vW_{x,j}^{r,\pm}(\tau)$ arising from its own estimation $\tau$, in \Cref{lem_nondef} and \Cref{lem_def}, respectively. The bound on the error probability of misclassified non-defective items (\Cref{lem_nondef}) substantially relies on the fact that we are able to compute a specific, non-trivial integral of Poisson probabilities explicitly, see \Cref{claim:antideriv}.
In the second step in \Cref{lemma_endgame_misclassified}, we bound the number of items for which these quantities differ, using a graph expansion argument.

\subsection{Cleaning Phase}

\Cref{prop_step1} states that, after the second stage of \SPOT, we have achieved an initial approximation of $\SIGMA$ with $o(k)$ misclassified items. 
The next and final step of \SPOT\ `cleans up' any mistake that the second step made, by locally repairing any inconsistencies.

Let $x\in V$ be an item and assume that the second neighbourhood of $x$ is decoded correctly. 
If $x$ is included in at least one test that contains exactly $t-1$ other defective items, we know the state of $x$ precisely, since the test outcome then depends solely on the state of $x$. We then say that there is a {\em pivotal} text for $x$.
In \Cref{Lem_distphixstar} in \Cref{Sec:Clean}, we show that the condition 
$m>(1+\eps)\minf$ from \cref{eq:m_inf} ensures that each item is contained in at least $\Omega(\sqrt \Delta)$ of such pivotal tests. With this notion of pivotal tests, we can use the same approach as phase 3 of the algorithm \SPIV\ for binary group testing in \cite{coja_spiv}.

\IncMargin{1em}
\begin{algorithm}[H]
  \KwData{$\tau$ from \SPOT, phase~2}
  \setcounter{AlgoLine}{10}
  Let $\tau^{(1)}=\tau$\;
  \For{ $i = 1,\dots, \lceil\log n\rceil$} {
  For all $x\in V[s+1]\cup\cdots\cup V[\ell]$ calculate\\
   $\displaystyle \qquad S_x(\tau^{(i)})=\sum_{a\in\partial x:\hat\SIGMA_a=1 \;}\vecone\cbc{\sum_{y\in\partial a\setminus\cbc x}\tau^{(i)}_y =  t-1}$\;
  Let 
  $\displaystyle\tau_x^{(i+1)}=
  	\begin{cases}
  	\tau_x^{(i)}&\mbox{ if }x\in V[1]\cup\cdots\cup V[s],\\
  	\vecone\cbc{S_x \bc{\tau^{(i)}}>\log^{1/4}n}&\mbox{ otherwise. }
  	\end{cases}$\;
  	}
  \KwRet{$\tau^{(\lceil\log n\rceil)}$}
\caption{\SPOT, phase~3.}\label{SC_algorithm}
\end{algorithm}
\DecMargin{1em}

In each cleaning iteration, all items that are contained in at least $\log^{1/4} n$ positive pivotal tests (according to the current approximation of $\SIGMA$) are diagnosed as defective. We will argue that both for non-defective and defective items $x$ that are misclassified at stage $j+1$, there are $\log^{1/4}n$ positive pivotal tests in its neighbourhood that contain at least one item that is misclassified at stage $j$. 
However, the same expansion property of $\G_{\mathrm{sc}}$ that was used in the last phase (\Cref{lemma_endgame_misclassified}) implies that only few items can have a very dense second neighbourhood of misclassified items, if the initial set is not too big. 
The following proposition indeed shows that the number of misclassified items is reduced by at least a constant factor in each iteration:

\begin{restatable}[]{proposition}{prop_endgame}
\label{prop_endgame}
Suppose that $(1+\eps)\minf\leq m=O(n^\theta\log n)$ and let $\tau=\tau^{(1)}$ be the output of phase 2 of \SPOT.
\Whp\ for all $1\leq j\leq\lceil\log n\rceil$,
\begin{align*}
\sum_{x\in V}\vecone\{\tau^{(j+1)}_x\neq \SIGMA_x\}\leq\frac13\sum_{x\in V}\vecone\{\tau^{(j)}_x\neq \SIGMA_x\}.
\end{align*}
As a result, \whp\ $\tau^{(\lceil\log n\rceil)} = \SIGMA$.
\end{restatable}
Obviously, \Cref{prop_endgame} completes the proof of \Cref{thm_alg_ex}. The proof of \Cref{prop_endgame} is provided in \Cref{Sec:Clean}.

\section{Proof of \Cref{prop_naive_thresholding} (Basic Thresholding)} \label{SSec:NaiveThersholding}
Recall that the \bth\ Algorithm (\Cref{Algo_Naive}) proceeds by counting the number of positive tests for each item and classifies it as defective if and only if the count exceeds a fixed threshold. We employ a constant column design where each item joins $\Delta' = c_{\mathrm{seed}} d' \log \frac{N}{K}$ tests independently with replacement. 

Recall the shorthands $P_{t-1}(d), P_{\sss \leq t-1}(d)$ from \cref{def_thr_para}. 
Intuitively, the quantities $P_{t-1}(d')$ and $ P_{\sss \leq t-1}(d')$ represent the asymptotic probabilities that in the set-up of Proposition~\ref{prop_naive_thresholding}, a test has exactly $t-1$ or at most $t-1$ defective items, respectively.

The following lemma compares the conditional probabilities of a positive test result based on the true state of one of its items $x$.

\begin{lemma} \label{Lem:probT}
    For all $K' \in [K-\sqrt{K}\log N, K + \sqrt{K}\log N]$, where $K = N^{\theta}(s/\ell)^{1-\theta}$, and any test $a \in \partial x$, with probability $1-o(N^{-2})$, the conditional probabilities are
    \begin{align*}
         \Pr(\hat{\SIGMA}_a' = 1 \mid  \partial x, \SIGMA_x=0, \vec \Gamma_a') &= 1- P_{\sss \leq t-1}(d')+o(1), \\
         \Pr(\hat{\SIGMA}_a' = 1 \mid  \partial x, \SIGMA_x=1, \vec \Gamma_a') &= 1-( P_{\sss \leq t-1}(d') - P_{t-1}(d'))+o(1).
    \end{align*}
\end{lemma}

\begin{proof}
Let $a \in \partial x$. Since given $\vGamma_a$, the $\vGamma_a$ half-edges of $a$ choose available item half-edges to items uniformly at random without replacement,
\begin{align}
    \Pr(\hat{\SIGMA}_a' = 1 \mid \partial x, \SIGMA_x=0, \vec \Gamma_a') = 1 - \sum_{r=0}^{t-1}\frac{\binom{\Delta'K'}{r}\binom{(N-K')\Delta'-1}{\vGamma_a'-1-r}}{\binom{N\Delta'-1}{\vGamma_a'-1}}.
\end{align}
Similarly,
\begin{align}
    \Pr(\hat{\SIGMA}_a' = 1 \mid \partial x, \SIGMA_x=1, \vec \Gamma_a') = 1 - \sum_{r=0}^{t-2}\frac{\binom{\Delta'K'-1}{r}\binom{(N-K')\Delta'-1}{\vGamma_a'-1-r}}{\binom{N\Delta'-1}{\vGamma_a'-1}}.
\end{align}
Recall that by Remark \ref{rem_seed_deg}, with probability $1-o(N^{-2})$, 
\begin{align*}
     \frac{\Delta' N}{m_{\mathrm{seed}}}-\sqrt{\frac{\Delta' N}{m_{\mathrm{seed}}}}\log N \leq \vGamma'_{\min} \leq \vGamma'_{\max} \leq \frac{\Delta' N}{m_{\mathrm{seed}}}+\sqrt{\frac{\Delta' N}{m_{\mathrm{seed}}}}\log N.
\end{align*}
On this event, using the choices of our parameters and Stirling's formula, 
\begin{align}
    \Pr(\hat{\SIGMA}_a' = 1 \mid \partial x, \SIGMA_x=0, \vec \Gamma_a') = 1 - \sum_{i=0}^{t-1} \frac{(d')^i e^{-d'}}{i!} + o(1) = 1 - P_{\sss \leq t-1}(d') + o(1)
\end{align}
and 
\begin{align}
    \Pr(\hat{\SIGMA}_a' = 1 \mid \partial x, \SIGMA_x=1, \vec \Gamma_a') = 1 - \sum_{i=0}^{t-2} \frac{(d')^i e^{-d'}}{i!} + o(1) = 1 - P_{\sss \leq 2-1}(d') + o(1).
\end{align}
\end{proof}

\Cref{Lem:probT} and the linearity of expectation immediately give the following:

\begin{corollary} \label{Lem:expT_replacement}
    The conditional expectations of the scoring function $\mathcal{T}_x = \sum_{a \in \partial x} \hat{\SIGMA}_{a}'$ satisfy
    \begin{align*} 
        \mathbb{E}[\mathcal{T}_x \mid \partial x, \SIGMA_x=0] &= \Delta'(1-P_{\sss \leq t-1}(d')) + o(\Delta'), \\ 
        \mathbb{E}[\mathcal{T}_x \mid \partial x, \SIGMA_x=1] &= \Delta'(1-P_{\sss \leq t-1}(d') + P_{t-1}(d')) + o(\Delta').
    \end{align*}
\end{corollary}

\begin{proof}[Proof of \Cref{prop_naive_thresholding}]
Let $\alpha \in (0,1)$. We first consider the auxiliary threshold 
\[ T_{\alpha} = \Delta' \bc{1-P_{\sss \leq t-1}(d') + \alpha P_{t-1}(d') }. \]

While the test outcomes $\hat{\SIGMA}_a'$ for $a' \in \partial x$ are not strictly independent (since tests compete for a fixed number of defective edges), given $\partial x$, they are negatively associated indicator random variables. By Theorem \ref{ChernNeg}, sums of negatively associated random variables obey identical Chernoff bounds as independent random variables.

Let $\eps_1=\tfrac{\alpha P_{t-1}(d')}{1-P_{\sss \leq t-1}(d')}$ and $\eps_2 = \frac{(1-\alpha)P_{t-1}(d')}{1-P_{\sss \leq t-1}(d')+P_{t-1}(d')}$. 
Since then $(1 + \eps_1 +o(1)) \Erw \brk{\cT_x \mid \partial x, \SIGMA_x=0} = T_{\alpha}$ and $(1 - \eps_2 +o(1)) \Erw\brk{\cT_x \mid \partial x, \SIGMA_x=1} = T_{\alpha}$, and using \Cref{lem_chernoff}, we find 
\begin{align*}
    & \Pr\bc{ \cT_x >  T_{\alpha} \mid \SIGMA_x=0} \leq \exp\bc{-\frac{\eps_1^2}{2 + 2\eps_1/3} \Erw\brk{\cT_x \mid \partial x, \SIGMA_x=0}+o(\Delta') } \qquad \text{and}  \\
    & \Pr\bc{\cT_x \leq  T_{\alpha}\mid \SIGMA_x=1} \leq \exp\bc{-\frac{\eps_2^2}{2} \Erw\brk{\cT_x \mid \partial x, \SIGMA_x=1}+o(\Delta') }.
\end{align*}
We next  make use of our choice of $d'$, according to which $1-P_{\sss \leq t-1}(d') = 1/2$, so that $\eps_1 = 2\alpha P_{t-1}(d')$ and $\eps_2 = \frac{2(1-\alpha)P_{t-1}(d')}{1+2P_{t-1}(d')}$. Recalling that $\Delta' = (1+\delta)d'c(1-\theta')\log N$, the expected numbers of
 false diagnoses of defective or non-defective items can be upper bounded as
\begin{align*}
    \Erw\brk{\sum_{x\in V_0}  \vecone\{\cT_x >  T_{\alpha}\} } \leq& N \exp\bc{-\bc{\frac{(\alpha  P_{t-1}(d'))^2}{1 + 2\alpha P_{t-1}(d') /3} (1+\delta)d'c(1-\theta')+o(1)}\log N  },  \\
   \Erw\left[\sum_{x\in V_1}  \vecone\{\cT_x \leq  T_{\alpha}\} \right] \leq& K' \exp\bc{-\bc{\frac{((1-\alpha) P_{t-1}(d')) ^2}{1+2 P_{t-1}(d')} (1+\delta)d'c(1-\theta')  +o(1)} \log N}.
\end{align*} 

In order to show that \whp, there are no misclassified defective and non-defective items, it is sufficient to show that both terms on the r.h.s. are $o(1)$. The choices of $\alpha=\alpha^*$ as in \cref{def_thr_para} and $c=c_{\mathrm{seed}}$ as in \cref{eq:mseed}
ensure that 
\begin{align}\label{eq_seed_aq}
    (\alpha^*)^2 c_{\mathrm{seed}} = \frac{2}{1-\theta'} \cdot \frac{1}{d'P_{t-1}(d')^2} \qquad \text{and} \qquad  (1-\alpha^*)^2 c_{\mathrm{seed}} = \frac{2\theta'}{1-\theta'} \cdot \frac{1}{d'P_{t-1}(d')^2}.
\end{align}
Substituting these into the above two exponents translates into
\begin{align*}
    \Erw\brk{\sum_{x\in V_0}  \vecone\{\cT_x >  T_{\alpha}\} } \leq& \exp\bc{\bc{1-\frac{2}{1 + 2\alpha P_{t-1}(d') /3} (1+\delta)+o(1)}\log N  },  \\
   \Erw\left[\sum_{x\in V_1}  \vecone\{\cT_x \leq  T_{\alpha}\} \right] \leq& \exp\bc{\bc{1-\frac{2}{1+2 P_{t-1}(d')} (1+\delta) +o(1)} \theta' \log N}.
\end{align*}
Finally, since $P_{t-1}(d') \leq P_{\sss \leq t-1}(d') = \frac{1}{2}$ and $\alpha^* \leq 1$, both $1 + 2\alpha P_{t-1}(d') /3\leq 2$ and $1+2 P_{t-1}(d') \leq 2$. 
Therefore, for $T=T_{\alpha^*}$,  with $m = (1+\delta) \frac{2(1+\sqrt{\theta})}{1-\sqrt{\theta}} \cdot \frac{1}{d' P_{t-1}(d')^2} K \log \frac{N}{K} $ tests, the \bth\ algorithm recovers the ground truth correctly \whp
\end{proof}

\section{Proof Approximate Recovery (Proof of \Cref{prop_step1})} \label{Sec:Approx}
We split the proof of \Cref{prop_step1} into two main parts.
First, \Cref{lem_nondef} provides a \whp\ upper bound on the number of misclassified non-defective items, if the algorithm had always access to the true labels $\SIGMA$. To counterbalance the assumption of perfect knowledge, we look at slightly larger deviations from the approximate means of the statistics of defective items.  
\Cref{lem_def} provides a similar \whp\ bound on the number of misclassified defective items. 
The proofs can be found in Section \ref{sec:lem_nondef} and \ref{sec:lem_def}, respectively.

\begin{restatable}[]{lemma}{lemnondef}
     \label{lem_nondef}
        Suppose that $m \geq (1+\eps) \minf$. \Whp, there exists a $\delta=\delta(\eps)>0$ such that
        \begin{align}\label{eq_lem_nondef}
             \sum_{s < i \leq \ell} \sum_{x\in V_0[i]} \vecone\cbc{ \bigcap_{1 \leq j \leq s} \bigcap_{0\leq r \leq t-1} \cbc{\vW_{x,j}^{r,+}(\SIGMA) \geq (1-2\zeta) W_{i,j}^{r,+}, \vW_{x,j}^{r,-}(\SIGMA) \leq (1+2\zeta) W_{i,j}^{r,-}  }} \leq k^{1-\delta}.
        \end{align}
    \end{restatable}

\begin{restatable}[]{lemma}{lemdef}
 \label{lem_def}
    Suppose that $m \geq (1+\eps) \minf$. \Whp~ 
    \begin{align} \label{eq_def_count}
       &   \sum_{s < i \leq \ell} \sum_{x\in V_1[i]}   \vecone
       \Big\{ \bigcup_{1 \leq j \leq s} \bigcup_{0\leq r \leq t-1} \Big\{ 
       \abs{\vW_{x,j}^{r,+}(\SIGMA) - W_{1,j}^{r,+}} > (\zeta/2) W_{1,j}^{r,+} \nonumber \\ & \qquad \qquad \qquad \qquad \qquad \qquad \qquad \quad\text{ or }\abs{\vW_{x,j}^{r,-}(\SIGMA) - W_{1,j}^{r,-}} >   (\zeta/2) W_{1,j}^{r,-}\Big\}\Big\} \nonumber \\
       &\leq   k\cdot \exp\bc{-\Omega\bc{\frac{\log n}{\bc{\log \log n}^{2t+3}}}}.
    \end{align}
\end{restatable}

As mentioned, Lemmas~\ref{lem_nondef} and \ref{lem_def} are not exactly what we ultimately need, since our algorithm will misclassify items along the way. 
Since our algorithm uses an approximation $\tau$ of $\SIGMA$ throughout, we will now show that this lack of perfect information can be absorbed into the additional slack that the previous two lemmas have left ($2\zeta$ instead of $\zeta$ in \Cref{lem_nondef}, and $\zeta/2$ instead of $\zeta$ in \Cref{lem_def}).
To achieve this we will use the following expansion property:

\begin{lemma}[Lemma 3.3 of \cite{coja2025noisy}] 
	Let $0<\alpha,\beta<1$ be such that $\alpha+\beta>1$.
	Then \whp{} for any $T \subset V$ of size $|T|\leq\exp(-\log^\alpha n)k$, in $\G_{\mathrm{sc}}$,
	\begin{align*}
		\abs{\cbc{x\in V:\sum_{a\in\partial x\setminus F[0]}\vecone\cbc{T\cap\partial a\setminus \cbc x\neq\varnothing}\geq\log^{\beta}n}}\leq\frac{|T|}{8\log\log n}.
	\end{align*}
\label{lemma_endgame_misclassified}
\end{lemma}

This lemma states that for a small set of items, the number of items that share many tests with them is even smaller. 
This property limits the effect of error propagations in the spatial coupling scheme $\G_{\mathrm{sc}}$, and allows us to complete the proof of \Cref{prop_step1}:

\begin{proof}[Proof of \Cref{prop_step1}]
    Let $i \in [\ell]$ and $M[i]$ be the items in compartment $i$ that step 2 of \SPOT\ misclassifies. 
    We claim that \whp, for all $i \in [\ell]$,
    \begin{align}
    \label{eq_miscl_proc}
        \abs{M[i]} \leq k \exp\bc{-\frac{\log n}{(\log\log n)^{2t+3}}}.
    \end{align}
    This implies the claim as the total number of compartments is $\ell = \sqrt{\log n}$. 

    We show \eqref{eq_miscl_proc} by induction on $h$, the induction hypothesis being that \eqref{eq_miscl_proc} is valid for all $i\leq h$. 
    
For $i \in [s]$, \Cref{prop_naive_thresholding} shows that $\abs{M[i]}=0$ \whp, which initialises the induction hypothesis.

    To advance the induction hypothesis, suppose the claim is true for all $i$ with $1 \leq i\leq h-1$, and we prove that it also holds for $h \geq s+1$.  
    We trace potential misclassifications back to three causes: Let $M_1[h]$ denote the set of defectives in compartment $h$ which do not look like defective items even if no mistakes were made in previous compartments, $M_2[h]$ the set of non-defectives that look like defective items even if no mistakes were made in previous compartments, and $M_3[h]$ the set of items whose neighbourhoods already contain a lot of misclassified items. 
    More precisely, we define
    
{\footnotesize
    \begin{align*}
        M_1[h] &= \cbc{x\in V_1[h]  : \exists 0 \leq r \leq t-1, j \in [s]:   
        \abs{\vW_{x,j}^{r,+}(\SIGMA)-W_{1,j}^{r,+}} \geq \frac{\zeta W_{1,j}^{r,+}}{2} \text{ or } 
        \abs{\vW_{x,j}^{r,-}(\SIGMA)-W_{1,j}^{r,-}} \geq \frac{\zeta W_{1,j}^{r,-}}{2}} , \\
        M_2[h] &= \cbc{x\in V_0[h] : \forall 0 \leq r \leq t-1, j \in [s]: 
        \vW_{x,j}^{r,+}(\SIGMA) > (1-2\zeta)W_{1,j}^{r,+} \text{ and } \vW_{x,j}^{r,-}(\SIGMA) < (1+2\zeta)W_{1,j}^{r,-} }, \\
        M_3[h] &= \cbc{x\in V[h] : \exists 0 \leq r \leq t-1, j \in [s]: 
        \abs{\vW_{x,j}^{r,+}(\SIGMA)-\vW_{x,j}^{r,+}(\tau)} > \frac{\zeta W_{1,j}^{r,+}}{4} \text{ or } \abs{\vW_{x,j}^{r,-}(\SIGMA)-\vW_{x,j}^{r,-}(\tau)} > \frac{\zeta W_{1,j}^{r,-}}{4}}.
    \end{align*}
    }
     
Note that $M[h]\subseteq M_1[h]\cup M_2[h] \cup M_3[h]$. 
Observe that $\vW_{x,j}^{r,+}(\SIGMA)=\vW_{x,j}^{r,+}(\tau)$ for $j=s$ and all $0 \leq r \leq t-1$. 
Furthermore, for $j \in [s-1]$ and $0\leq r \leq t-1$, 
\begin{align} \label{eq_lower_W}
  \zeta W^{r,\pm}_{1,j} &\geq \zeta \frac{\Delta}{s}\pr\bc{\Po\bc{\frac{d^*(s-j)}{s}} = r}\min\cbc{\pr\bc{\Po\bc{\frac{d^*j}{s}}= t-r}, \pr\bc{\Po\bc{\frac{d^*j}{s}}= t-r-1} } \nonumber \\
  & = \zeta \frac{\Delta}{s}\eul^{-\frac{d^*(s-j)}{s}-\frac{d^*j}{s}} \frac{\bc{\frac{d^*(s-j)}{s}}^r}{r!}\min\cbc{ \frac{\bc{\frac{d^*j}{s}}^{t-r}}{(t-r)!}, \frac{\bc{\frac{d^*j}{s}}^{t-r-1}}{(t-r-1)!}} \nonumber \\
   &= \Omega\bc{\zeta \frac{\Delta}{s}\bc{\frac{j}{s}}^{t-r}\bc{1-\frac{j}{s}}^r } = \Omega\bc{\zeta \frac{\Delta}{s^{t+1}} } = \Omega\bc{\bc{\log n}^{5/8}}.
\end{align}
Thus, for each $x \in V[h]$, there exist $j \in [s-1]$ and $0\leq r \leq t-1$ such that $|\vW_{x,j}^{r,\pm}(\SIGMA)-\vW_{x,j}^{r,\pm}(\tau)| \geq \sqrt{\log n}$. As a consequence, $x$ has at least $\sqrt{\log n}$ neighbours that have non-empty intersection with the misclassified items from the previous compartments. 
Using the induction hypothesis and \Cref{lemma_endgame_misclassified}, we see that
$$\abs{M_3[h]} \leq \frac{M[h-s+1]\cup \dots \cup M[h-1]}{8\log\log n}\leq \frac{s k \exp\bc{-\frac{\log n}{\log\log^{2t+3} n}}}{8\log\log n} \leq k \exp\bc{-\frac{\log n}{\log\log^{2t+3} n}}/8 \;.$$
By \Cref{lem_nondef} and \Cref{lem_def}, we also know that
$$\abs{M_1[h]}+\abs{M_2[h]}\leq k\exp\bc{-\frac{\log n}{\log\log^{2t+3} n}}.$$
This concludes the induction step, and completes the proof.
\end{proof}

    \subsection{Asymptotics of $\Erw\brk{\vW_j}$}

    The first step to prove \Cref{lem_nondef} and \Cref{lem_def} is to justify our approximation of the expectation of $\vW_{x,j}^{r,\pm}$ of the number of different kinds of tests that an item in compartment $i$ can join. 
    To that end we first determine the true distribution of the vector $\vW_{x,j}(\SIGMA)\in \mathbb{N}_0^{2(t+1)},$ given by

    \begin{align} \label{def_wvec}
			\vW_{x,j}(\SIGMA) =
			 \begin{pmatrix}
				\vW_{x,j}^{0,-}(\SIGMA) \\
				\vdots \\
				\vW_{x,j}^{t-1,-}(\SIGMA) \\
				\vW_{x,j}^{0,+}(\SIGMA) \\
				\vdots \\
				\vW_{x,j}^{t-1,+}(\SIGMA) \\
				\vW_{x,j}^{\geq t,+}(\SIGMA)
			\end{pmatrix}.
		\end{align}
    Let
    \begin{align}\label{def_vk}
    \vec k = \bc{\vec k[i]}_{i=0}^{\ell} :=  \bc{\abs{V_1[i]}}_{i=1}^{\ell}
\end{align}
be the vector of the numbers of defective items per compartment. 
For all $i \in [\ell], j \in [s]$, let
\begin{align}
  V_1[i-(s-j);i-1] :=   V_1[i-(s-j)] \cup \ldots \cup V_1[i-1].
\end{align}

\begin{definition} \label{def_ps}
Fix $s<i\leq \ell$ and $j \in [s]$, and, for $a\leq b$, let 
    \begin{equation}
    \label{k[a;b]-def}
    \vec k[a;b]:=\sum_{u=a}^{b}\vec k[u].
    \end{equation}
For $r \geq 0$, set
    \begin{align}
        \vec p_{0,j,i}^{r,+} &= \Pr\bc{\Bin\bc{\vec k[i-(s-j);i-1]\frac{\Delta}{s}, \frac{\ell}{m}} = r~\Big\vert~ \vk}\nonumber\\
        &\qquad\times\Pr\bc{\Bin\bc{\vec k[i;i+j-1]\frac{\Delta}{s}, \frac{\ell}{m}} \geq t-r~\Big\vert~\vk},
        \label{p-0r+-def}\\
		\vec p_{0,j,i}^{r,-} &= \Pr\bc{\Bin\bc{\vec k[i-(s-j);i-1]\frac{\Delta}{s}, \frac{\ell}{m}} = r~\Big\vert~ \vk}\nonumber\\
        &\qquad\times\Pr\bc{\Bin\bc{\vec k[i;i+j-1]\frac{\Delta}{s}, \frac{\ell}{m}} < t-r~\Big\vert~ \vk},\\
        \vec p_{1,j,i}^{r,+} &= \Pr\bc{\Bin\bc{\vec k[i-(s-j);i-1]\frac{\Delta}{s}, \frac{\ell}{m}} = r~\Big\vert~ \vk}\nonumber \\
            &\qquad\times\Pr\bc{\Bin\bc{\vec k[i;i+j-1]\frac{\Delta}{s}-1, \frac{\ell}{m}} \geq t-r-1~\Big\vert~ \vk},  \\
		\vec p_{1,j,i}^{r,-} &= \Pr\bc{\Bin\bc{\vec k[i-(s-j);i-1]\frac{\Delta}{s}, \frac{\ell}{m}} = r~\Big\vert~ \vk}\nonumber \\
            &\qquad\times\Pr\bc{\Bin\bc{\vec k[i;i+j-1]\frac{\Delta}{s}-1, \frac{\ell}{m}} < t-r-1~\Big\vert~ \vk}, \label{p-1r--def}\\
        \vec p_{j,i}^{\geq t} &= \Pr\bc{\Bin\bc{\vec k[i-(s-j);i-1]\frac{\Delta}{s}, \frac{\ell}{m}} \geq t~\Big\vert~ \vk}.
        \label{p-geqt-def}
    \end{align}
 The random variable $\vec p_{b,j,i}^{r,\pm}$ describes the conditional probability that a given test in compartment $i+j-1$ that is connected to item $x$ in compartment $i$ with $\SIGMA_x=b$ has $r$ edges to defective items in compartments $i+j-s$ to $i-1$ and has a  positive or negative ($\pm$) outcome. We further define
    \begin{align*}
			\vec \vp_{1,j,i} &=(\vp_{1,j,i}^{0,-}, \dots, \vp_{1,j,i}^{t-1,-},\vp_{1,j,i}^{0,+}, \dots, \vp_{1,j,i}^{t-1,+}, \vp_{j,i}^{\geq t}), \quad	\vp_{0,j,i} = (\vp_{0,j,i}^{0,-}, \dots, \vp_{0,j,i}^{t-1,-},\vp_{0,j,i}^{0,+}, \dots, \vp_{0,j,i}^{t-1,+}, \vp_{j,i}^{\geq t}), \\
             \vp_{1,i} &= (\vp_{1,1,i}, \dots, \vp_{1,s,i}), \quad \text{and} \quad	\vp_{0,i} = (\vp_{0,1,i}, \dots, \vp_{0,s,i}).
		\end{align*}
\end{definition}

\begin{restatable}[Law of $\vW_{x,j}$]{claim}{claim_expextation}
\label{claim_expextation}
    Let $x \in V[i]$ for $s<i\leq \ell$. Conditionally on $\vec k$ and $\SIGMA_x = 0$, the random vectors $\vW_{x,j}$ in \eqref{def_wvec} satisfy
		\begin{align}
			\vW_{x,j}(\SIGMA)
        \sim \Mult\left(\Delta/s, \vp_{0,j,i}\right),
		\end{align}
         while, conditionally on $\vec k$ and $\SIGMA_x = 1$, 
		\begin{align}
			\vW_{x,j}(\SIGMA) 
        \sim \Mult\left(\Delta/s, \vp_{1,j,i}\right).
		\end{align}
\end{restatable}

\begin{proof}
    Item $x$ independently draws $\Delta/s$ balls from compartment $i+j-1$ with replacement. For $\SIGMA_x=b$, the (conditional) probability that any test connects to exactly $r$ defective items in compartments $i+j-s$  to $i-1$ and has a  positive or negative ($\pm$) outcome is equal to the probabilities $\vec p_{b,j,i}^{r,\pm}$ as introduced in Definition~\ref{def_ps}: The probability that any test connects to exactly $r$ defective items in compartments $i+j-s$  to $i-1$ is given by the binomial probability $\Pr\bc{\Bin\bc{\vec k[i-(s-j);i-1]\frac{\Delta}{s}, \frac{\ell}{m}} = r\mid \vk}$, where $ \vec k[i-(s-j);i-1]\Delta/s$ is the number of edges emanating from defective items, and $\ell/m$ is the probability that one of those edges connects to a particular test. 
    
    If $r\leq t-1$, then the probability that this test is positive or negative depends on the number of edges that connect it to defective items in compartments $i$ to $i+j-1$, excluding the edge connecting $x$ to the test. 
    If the item is defective, then the test is positive (negative, respectively) if at least (at most, respectively) $t-r-1$ of the $\vec k[i;i+j-1]\frac{\Delta}{s}-1$ other edges connect to defective items; if it is negative, then the test is positive (negative, respectively) if at least (at most, respectively) $t-r$ of the $\vec k[i;i+j-1]\frac{\Delta}{s}$ other edges connect to defective items.
\end{proof}

Recall the probabilities introduced in \eqref{eq_def_q0p}--\eqref{eq_def_q1m}, which will act as the limits of the conditional probabilities in \eqref{p-0r+-def}--\eqref{p-1r--def}. Moreover, for each $1\leq j \leq s$, let
 \begin{align*}
			q_j^{\geq t} = q_{0,j}^{\geq t,+} = q_{1,j}^{\geq t,+} = \Pr\bc{\Po\bc{\frac{d^*(s-j)}{s}} \geq t} \,,
		\end{align*}
which will act as the limit of $\vec p_{j,i}^{\geq t}$ in \eqref{p-geqt-def}.
To summarise the notation, we define 
		\begin{align*}
			q_{1,j} &= (q_{1,j}^{0,-}, \dots, q_{1,j}^{t-1,-},q_{1,j}^{0,+}, \dots, q_{1,j}^{t-1,+}, q_j^{\geq t}), \quad	q_{0,j} = (q_{0,j}^{0,-}, \dots, q_{0,j}^{t-1,-},q_{0,j}^{0,+}, \dots, q_{0,j}^{t-1,+}, q_j^{\geq t}). 
		\end{align*}
We next show that the conditional probabilities $\vec p$ in \eqref{p-0r+-def}--\eqref{p-geqt-def} are well approximated by their associated $q$:
\begin{claim}[Limits of conditional probabilities]\label{claim_distpq}
    Conditionally on $\cE$, for  all $1 \leq i \leq \ell, 0 \leq r < t, 1 < j \leq s,$
    \begin{align}
        \vec p_{0,j, i}^{r,+} &= q_{0,j}^{r,+} ( 1+ O(\log^3 n / \sqrt{k})),\\
        \vec p_{0,j, i}^{r,-} &= q_{0,j}^{r,-} ( 1+ O(\log^3 n / \sqrt{k})),  \\
        \vec p_{1,j, i}^{r,+} &= q_{1,j}^{r,+}  ( 1+ O(\log^3 n / \sqrt{k})),\\
        \vec p_{1,j, i}^{r,-} &= q_{1,j}^{r,-}  ( 1+ O(\log^3 n / \sqrt{k})),\\
        \vec p_{i,j}^{\geq t} &= q_{j}^{\geq t}  ( 1+ O(\log^3 n / \sqrt{k})).
    \end{align}
\end{claim}
The proof of \Cref{claim_distpq} consists of routine Poisson approximations, and is deferred to \Cref{Appx:disqp}.

    \subsection{Misclassified Non-Defective Items (Proof of \Cref{lem_nondef})}
\label{sec:lem_nondef}

In this section, we prove \Cref{lem_nondef}, and show that the probability that the statistics of a non-defective item closely match those of a defective one is small, under the assumption that we have {\em perfect} knowledge of the ground truth on the previous compartments. 
For this, we recall that the {\em Kullback-Leibler divergence} of \invisible{$p,q\in(0,1)$ is denoted by
\begin{align} \label{eq:kl}
    \KL{q}{p} = q \log \bc{\frac{q}{p}} + (1-q) \log \bc{\frac{1-q}{1-p}}.
\end{align}
For two vectors}two probability mass functions $p,q \in [0,1]^n$
\invisible{with $\norm{p}_1 = \norm{q}_1 = 1$, the Kullback-Leibler divergence}is defined as
\begin{align}
    \KL{q}{p} = \sum_{i=1}^n q_i \log \bc{\frac{q_i}{p_i}},
\end{align}
and we let
		\begin{align}\label{optpos}
			\cM_{s}  = \frac 1 s \sum_{j=1}^{s}\KL{q_{1,j} }{q_{0,j}}.
		\end{align} 

Below, we use the abbreviations $P_k(\mu)=\pr\bc{\Po(\mu)=k}, P_{\sss \leq k}(\mu)=\pr\bc{\Po(\mu)\leq k}$ and $P_{\sss > k}(\mu)=\pr\bc{\Po(\mu)>k}$.

\begin{claim}[Convergence of misspecification rate of non-defective items]\label{lem_calc}
Recall $\cM_{s}$ from \eqref{optpos}. Then
    \begin{align*}
       \liminf_{s\rightarrow \infty} \cM_{s} \geq 
       \int_0^1 \sum_{r=0}^{t-1} P_r(d^*(1-z)) \KL{P_{\sss\geq t-r-1}(d^*z) }{P_{\sss \geq t-r}(d^*z)} \, \mathrm{dz}.
    \invisible{\frac{H(\Pr\bc{\Po(d^*) \geq t }  ) }{d^*}}
    \end{align*}
\end{claim}
\begin{proof}
      Fix $\eps > 0$. Then, $\eps s\geq 1$ for $s$ large enough, and, since all summands are non-negative,
       we can neglect the first $\eps s$ terms, to obtain
    \begin{align}\label{eq_lower_M}
        \liminf_{s \to \infty} \cM_{s} \geq \liminf_{s \to \infty} \frac{1}{s} \sum_{j = \eps s}^s \KL{q_{1,j} }{q_{0,j}}.
    \end{align}
Thanks to the truncation, the r.h.s.\ of \eqref{eq_lower_M} can be seen as the Riemann sum of a bounded function over a bounded interval. We thus get
 \begin{align}
        \liminf_{s \to \infty} \cM_{s} \geq  \int_\eps^1 \sum_{r=0}^{t-1} P_r(d^*(1-z)) \KL{P_{\sss\geq t-r-1}(d^*z) }{P_{\sss \geq t-r}(d^*z)} \, \mathrm{dz}.
    \end{align}
Letting $\eps\searrow 0$, and using the fact that the integrand is non-negative, proves the claim.
\end{proof}

In order to identify the limiting integral in Claim \ref{lem_calc}, the following result is crucial:
\begin{claim}[Analysis of the integrand] \label{claim:antideriv}
For every $d$,
    \begin{equation}
       \sum_{r=0}^{t-1} P_r(d(1-z)) \KL{P_{\sss\geq t-r-1}(dz) }{P_{\sss \geq t-r}(dz)}= \frac{\partial}{\partial z}\;\frac{1}{d}\sum_{r=0}^{t-1} P_r(d(1-z)) H\bc{P_{\sss \geq t-r}(dz)}.
    \end{equation}
\end{claim}

\begin{proof} We use the product rule, as well as $\frac{\partial}{\partial \mu}P_l(\mu)=P_{l-1}(\mu)-P_l(\mu)$, and thus, $\frac{\partial}{\partial \mu}P_{\sss \geq l}(\mu)=P_{l-1}(\mu)$, to compute the derivative of $\frac{1}{d}\sum_{r=0}^{t-1} P_r(d(1-z)) H\bc{P_{\sss \geq t-r}(dz)}$ as
\begin{align*}
     &\frac{1}{d} \sum_{r=0}^{t-1} d\bc{ - P_{r-1}(d(1-z)) + P_r(d(1-z))} H\bc{P_{\sss \geq t-r}(dz)} \\
     &\qquad + P_r(d(1-z))\log\bc{\frac{P_{\sss \geq t-r}(dz)}{P_{\sss < t-r}(dz)}}d\bc{P_{\sss < t-r-1}(dz) -  P_{\sss < t-r}(dz)}  \\
     &= \sum_{r=0}^{t-1} \bc{ - P_{r-1}(d(1-z)) + P_r(d(1-z))} H\bc{P_{\sss \geq t-r}(dz)} \\
     &\qquad - P_r(d(1-z))\log\bc{\frac{P_{\sss \geq t-r}(dz)}{P_{\sss < t-r}(dz)}} P_{t-r-1}(dz).
     \end{align*}
We simplify the r.h.s.\ as
    \begin{align*}
     &\sum_{r=0}^{t-1} P_r(d(1-z)) \bigg( -P_{\sss \geq t-r}(dz) \log\bc{P_{\sss \geq t-r}(dz)} - P_{\sss < t-r}(dz) \log\bc{P_{\sss < t-r}(dz)} \\
     &- P_{t-r-1}(dz) \log\bc{P_{\sss \geq t-r}(dz)} + P_{t-r-1}(dz) \log\bc{P_{\sss < t-r}(dz)} \bigg)- P_{r-1}(d(1-z))H\bc{P_{\sss \geq t-r}(dz)}.
     \end{align*}
We shift the sum over $r$ in the last term by 1, use that $H\bc{P_{\sss \geq 0}(dz)}=0$ so that no boundary term appears, and recombine, to rewrite this expression as
     \begin{align*}
     &\sum_{r=0}^{t-1} P_r(d(1-z))\\
     &\qquad
     \bigg( -P_{\sss \geq t-r-1}(dz) \log\bc{P_{\sss \geq t-r}(dz)} - P_{\sss < t-r-1}(dz) \log\bc{P_{\sss < t-r}(dz)}- P_r(d(1-z))H\bc{P_{\sss \geq t-r-1}(dz)}\bigg).
     \end{align*}
Using the definition of $H\bc{P_{\sss \geq t-r-1}(dz)}$, we can rewrite
    \begin{align*}
    &-P_{\sss \geq t-r-1}(dz) \log\bc{P_{\sss \geq t-r}(dz)} - P_{\sss < t-r-1}(dz) \log\bc{P_{\sss < t-r}(dz)}- P_r(d(1-z))H\bc{P_{\sss \geq t-r-1}(dz)}\\
     &=-P_{\sss \geq t-r-1}(dz) \log\bc{P_{\sss \geq t-r}(dz)} - P_{\sss < t-r-1}(dz) \log\bc{P_{\sss < t-r}(dz)} \\
     &\qquad +  P_{\sss \geq t-r-1}(dz) \log\bc{P_{\sss \geq t-r-1}(dz)} +P_{\sss \geq t-r-1}(dz) \log\bc{P_{\sss \geq t-r-1}(dz)}\\
     &=\KL{P_{\sss \geq t-r-1}(dz)}{P_{\sss \geq t-r}(dz)}.
\end{align*}
This completes the proof.
\end{proof}

\begin{claim}[Convergence of misspecification rate of non-defective items (cont.)]\label{lem_calc_cont}
Recall $\cM_{s}$ from \eqref{optpos}. Then
    \begin{align*}
       \liminf_{s\rightarrow \infty} \cM_{s} \geq 
       \frac{H(\Pr\bc{\Po(d^*) \geq t }  ) }{d^*}.
    \end{align*}
\end{claim}

\begin{proof}
By \Cref{claim:antideriv}, for any $\eps >0$,
    \begin{align*}
  &  \int_\eps^1\sum_{r=0}^{t-1} P_r(d^*(1-z)) \KL{P_{\sss\geq t-r-1}(d^*z) }{P_{\sss \geq t-r}(d^*z)} \, \mathrm{dz} 
        = \brk{ \frac{1}{d^*}\sum_{r=0}^{t-1} P_r(d^*(1-z)) H\bc{P_{\sss \geq t-r} (d^*z)}}_{\eps}^1\\
        &= \frac{1}{d^*} H(P_{\sss \geq t}(d^*)) - \frac{1}{d^*}\sum_{r=0}^{t-1} P_r(d^*(1-\eps)) H\bc{P_{\sss \geq t-r}(d^*\eps)}.
    \end{align*}
 Taking $\lim_{\eps \searrow 0}$, and using that the sum vanishes as $\eps \searrow 0$ together with \Cref{lem_calc}, we obtain
   \begin{align*}
      &  \liminf_{s\rightarrow \infty} \cM_{s} \geq
        \frac{1}{d^*} H(P_{\sss \geq t}(d^*))=\frac{1}{d^*} H(\Pr\bc{\Po(d^*) \geq t}). 
    \end{align*} 
 \invisible{Thus for all $\eps >0$,
  \begin{align*}
        \liminf_{s\rightarrow \infty} \cM_{s} \geq \frac{1}{d^*} H(\Pr\bc{\Po(d^*) \geq t}) - \frac{1}{d^*}\sum_{r=0}^{t-1} \Pr\bc{\Po(d^*(1-\eps)) = r} H\bc{\Pr\bc{\Po(d^*\eps) \geq t-r}}.
    \end{align*}
}
\end{proof}

We next bound the event that a non-defective item satisfies all classification conditions for a defective item, so that it will initially be misclassified:
    \begin{claim}[Bound on misclassification probability of non-defective items]\label{claim_pnondef}
For $s<i\leq \ell$, $j \in [s], 0 \leq r \leq t-1$ and $x \in V[i]$, let
        \begin{align*}
            \fE_{x,j,r} = \cbc{\vW_{x,j}^{r,+}(\SIGMA) \geq (1-2\zeta) W_{1,j}^{r,+} \text{ and } \vW_{x,j}^{r,-}(\SIGMA) \leq (1+2\zeta) W_{1,j}^{r,-} }.
        \end{align*}
        Then 
		\begin{align*}
			&\pr\bc{\bigcap_{j\in [s]} \bigcap_{r\in \{0,\dots,t-1\}}  \fE_{x,j,r} ~\Big\vert~ x\in V_0[i]} \leq\exp(-(1+o(1)) \Delta\cM).
		\end{align*}
	\end{claim}
\begin{proof}

Recall that $\cE$ is the event that properties 1.-3. from \Cref{Lemma_GammaMinMax} hold.
Given $\vk$, the events $\bigcap_{r \in \cbc{0, \dots, t-1}}\fE_{x,j,r}$ are {\em independent} for different $j$, as they depend on disjoint sets of  tests. Thus,
    \begin{align}
		\label{misclassification-healthy}\Pr&\bc{\bigcap_{j\in [s]} \bigcap_{r\in \{0,\dots,t-1\}}
         \fE_{x,j,r} ~\Big\vert~ \cE, \,\vk,\,x\in V_0[i]}
            = \prod_{j=1}^s \pr\bc{ \bigcap_{r\in \{0,\dots,t-1\}}  \fE_{x,j,r} ~\Big\vert~ \cE,\vk, \,x\in V_0[i] }.
		\end{align}
    We next write the last equation in vector form. For this, define the $d$-dimensional rectangles
\begin{align*}
    C_j^{2\zeta, -} &:= \cbc{0, \ldots, \lceil\bc{1+2\zeta}W_{1,j}^{0,-}\rceil} \times \ldots \times \cbc{0, \ldots, \lceil\bc{1+2\zeta}W_{1,j}^{t-1,-}\rceil}, \\
    C_j^{2\zeta, +} &:= \cbc{\lfloor\bc{1-2\zeta}W_{1,j}^{0,+}\rfloor, \ldots, \frac{\Delta}{s}} \times \ldots \times \cbc{\lfloor \bc{1-2\zeta}W_{1,j}^{t-1,+}\rfloor, \ldots, \frac{\Delta}{s}}.
\end{align*}
 and set $C_j^{2\zeta}:= \cbc{w \in C_j^{2\zeta, -} \times C_j^{2\zeta, +} \times [0,\infty)\colon  \sum_{u=1}^{2t+1}w_u = \Delta/s}$. Observe that for $j=s$ and $r \in [t-1]$, $W_{1,j}^{r,-} = W_{0,j}^{r,+} = 0$, so that the corresponding coordinates are restricted to $0$.
 Additionally, recall the definition \eqref{def_wvec} of $\vW_{x,j}(\SIGMA)$.
 Using this notation, we can rewrite the probability on the r.h.s.\ of \eqref{misclassification-healthy} as

\begin{align*}
    \prod_{j=1}^s   \pr\bc{ \bigcap_{r\in \{0,\dots,t-1\}}  \fE_{x,j,r} \mid \cE, \,\vk,\,x\in V_0[i] }
         = & \prod_{j=1}^s \pr\bc{ \vW_{x,j}(\SIGMA) \in C_j^{2\zeta}   ~\Big\vert~\cE, \,\vk,\,x\in V_0[i] } 
         \\=&  \prod_{j=1}^s \sum_{w \in C_j^{2\zeta} } \pr\bc{ \vW_{x,j}(\SIGMA) = w  \mid \cE, \,\vk,\,x\in V_0[i] }.
\end{align*}
Note that $|C_j^{2\zeta}| = O\bc{(\Delta/s)^{2t+1}}$ as each component of $w \in C_j^{2\zeta}$ can take at most $O(\Delta/s)$ distinct values. We can thus estimate the sum by the maximum over all elements in $C_j^{2\zeta}$ as
\begin{align}
   & \prod_{j=1}^s \sum_{w \in C_j^{2\zeta} } \pr\bc{ \vW_{x,j}(\SIGMA) = w  \mid \cE, \,\vk,\,x\in V_0[i]} \nonumber\\
             \leq & O\bc{\bc{\frac{\Delta}{s}}^{s(2t+1)}} \prod_{j=1}^s \max_{w_j=(w_j^-,w_j^+,w_j') \in C_j^{2\zeta}} \pr\bc{ \vW_{x,j}^+(\SIGMA) = w_j^+, \vW_{x,j}^-(\SIGMA) = w_j^-  \mid \cE, \,\vk,\,x\in V_0[i] } \nonumber \\
             = &\exp\bc{o\bc{\Delta}} \prod_{j=1}^s \max_{w_j=(w_j^-,w_j^+,w_j') \in C_j^{2\zeta}}   \pr\bc{ \vW_{x,j}^+(\SIGMA) = w_j^+, \vW_{x,j}^-(\SIGMA) = w_j^-  \mid \cE, \,\vk,\,x\in V_0[i]},
\end{align}
 where $w_j=(w_j^-,w_j^+,w_j')\in C_j^{2\zeta}$ indicates that $w_j^{\pm} \in C_j^{2\zeta, \pm}$,
 and we have used that $(\Delta/s)^{s(2t+1)}=\exp{(o\bc{\Delta})}$. 
 By Claim \ref{claim_expextation}, and Stirling's formula for the asymptotics of the multinomial probabilities,
\begin{align}
    \exp\bc{o\bc{\Delta}} &\prod_{j=1}^s \max_{w_j=(w_j^-,w_j^+,w_j') \in C_j^{2\zeta}}   \pr\bc{ \vW_{x,j}^+(\SIGMA) = w_j^+, \vW_{x,j}^-(\SIGMA) = w_j^-  \mid \cE, \,\vk,\,x\in V_0[i] } \nonumber \\
            & \leq \exp\bc{o\bc{\Delta}} \prod_{j=1}^s \max_{w_j \in C_j^{2\zeta}}  \exp\bc{-\frac{\Delta}{s} \KL{\frac{w_j}{\Delta/s}}{\vec p_{0,j,i}} + O\bc{\log \frac{\Delta}{s}}} \nonumber \\
            & \leq \exp\bc{o\bc{\Delta}}  \max_{w=(w_1,\ldots,w_s) \in \prod_{j=1}^{s} C_j^{2\zeta}}   \exp\bc{-\frac{\Delta }{s} \sum_{j=1}^s \KL{\frac{w_j}{\Delta/s}}{\vec p_{0,j,i}} }.
\end{align}
Recall that $\vec p_{0,s,i}^{r,\pm} = q_{0,s}^{r,\pm} = 0$ for $r \in [t-1]$, in which case also by our definition of the rectangles also $w_s^{r,\pm} = 0$, or in all other cases, $\vec p_{0,s,i}^{r,\pm}>0$ conditionally on $\cE$ and $q_{0,j}^{r, \pm} = \Omega(s^{-t})$, as demonstrated in \eqref{eq_lower_W}. As $\KL{a_n}{(1+\eps_n)b_n} = \KL{a_n}{b_n} - \sum_{i \in [2t+1]}a_n^{(i)}\log\bc{1+\eps_n}$ and the coordinates of each $w_j$ sum to $\Delta/s$, the asymptotics of $\vec p_{0,j,i}$ in Claim \ref{claim_distpq} entail that, conditionally on $\cE$,
\begin{align}
 \exp&\bc{o\bc{\Delta}}  
    \max_{(w_1,\ldots,w_s) \in \prod_{j=1}^{s} C_j^{2\zeta}}   \exp\bc{-\frac{\Delta }{s} \sum_{j=1}^s \KL{\frac{w_j}{\Delta/s}}{\vec p_{0,j,i}} } \nonumber \\
            & \leq \exp\bc{o\bc{\Delta}}  \max_{(w_1,\ldots,w_s) \in  \prod_{j=1}^{s} C_j^{2\zeta}}   \exp\bc{-\frac{\Delta }{s} \sum_{j=1}^s \KL{\frac{w_j}{\Delta/s}}{\bc{1+O\bc{\frac{\log^3 n}{\sqrt{k}}}}q_{0,j}} } \nonumber  \\
            &= \exp\bc{o\bc{\Delta}}  \max_{(w_1,\ldots,w_s) \in  \prod_{j=1}^{s} C_j^{2\zeta}}   \exp\bc{\Delta \log\bc{1+O\bc{\frac{\log^3 n}{\sqrt{k}}}} -\frac{\Delta }{s} \sum_{j=1}^s \KL{\frac{w_j}{\Delta/s}}{q_{0,j}} } \nonumber \\
            &=\exp\bc{o\bc{\Delta}}   \exp\bc{ -\Delta  \min_{(w_1,\ldots,w_s) \in  \prod_{j=1}^{s} C_j^{2\zeta}} \frac{1 }{s} \sum_{j=1}^s \KL{\frac{w_j}{\Delta/s}}{q_{0,j}} }.
\end{align}
This leads us to consider the optimisation problem
\begin{align}
    \cM_{s, \xi} &= \min_{w=(w_1,\ldots,w_s)} \frac{1}{s} \sum_{j=1}^s \KL{\frac{w_j}{\Delta/s}}{q_{0,j}}\nonumber\\
    &  \text{s.t.} \quad \forall j\in [s] \forall r=0,\ldots, t-1: w_{j}^{r,-} \leq \bc{1+\xi}W_{1,j}^{r,-}, \quad w_{j}^{r,+} \geq \bc{1-\xi}W_{1,j}^{r,+}.
\end{align}
We next compare $\cM_{s, 2\zeta}$ and $\cM_{s, 0}=\cM_{s}$.

Let $v$ be a feasible minimiser of $\frac{1}{s} \sum_{j=1}^s \KL{\frac{w_j}{\Delta/s}}{q_{0,j}}$ satisfying the constraints of $\cM_{s, 2\zeta}$, such that there exist $j \in [s], r \in \{0, \ldots, t-1\}$ such that $v_{j}^{r,-} > W_{1,j}^{r,-}$  or $v_{j}^{r,+} < W_{1,j}^{r,+}$. 
Then, by changing those coordinates to $W_{1,j}^{r,-}$ or $W_{1,j}^{r,+}$, respectively, and shifting the remaining discrepancy to the last coordinate $v_{s}^{\geq t}$, we obtain another vector $v'$. 
Since this changes each coordinate of the rescaled vector $v/(\Delta/s)$ by at most $O(\zeta)$ and the functions $w \mapsto  
\KL{\frac{w_j}{\Delta/s}}{q_{0,j}}$, $j \in [s]$, are equicontinuous on the compact interval $[0,1]$, we obtain
\begin{align}
    \frac{1}{s} \sum_{j=1}^s \KL{\frac{v_j}{\Delta/s}}{q_{0,j}} \geq \frac{1}{s} \sum_{j=1}^s \KL{\frac{v_j'}{\Delta/s}}{q_{0,j}} + o(1)
\end{align}
uniformly for all $v_1, \ldots, v_s$ and $v_1', \ldots, v_s'$ Thus, 
\begin{align}
    &\exp\bc{o\bc{\Delta}}   \exp\bc{ -\Delta  \min_{w=(w_1,\ldots,w_s) \in  \prod_{j=1}^{s} C_j^{2\zeta}} \frac{1 }{s} \sum_{j=1}^s \KL{\frac{w_j}{\Delta/s}}{q_{0,j}} }\\
     &=\exp\bc{o\bc{\Delta}}    \exp\bc{ -\Delta \min_{w=(w_1,\ldots,w_s) \in \prod_{j=1}^{s} C_j^{0}} \frac{1 }{s} \sum_{j=1}^s \KL{\frac{w_j}{\Delta/s}}{q_{0,j}} }.
\end{align}

Consider the optimisation problem $\cM_{s,0}=\cM_{s}$.
The KL-divergences are convex and reach their minimum at $w_j/\frac{\Delta}{s}=q_{0,j}$. 
Due to the constraints we have $q^{+,r}_{0,j}\leq q^{+,r}_{1,j} \leq w^{+,r}_j/\frac{\Delta}{s}$ and $q^{-,r'}_{0,j}\geq q^{-,r'}_{1,j} \geq w^{-,r'}_j/\frac{\Delta}{s}$. 
By convexity, the minimum is achieved at $w_j/\frac{\Delta}{s} = q_{1,j}$. This leads us to
\begin{align*}
    &\exp\bc{o\bc{\Delta}}    \exp\bc{ -\Delta \min_{w=(w_1,\ldots,w_s) \in \prod_{j=1}^{s} C_j^{0}} \frac{1 }{s} \sum_{j=1}^s \KL{\frac{w_j}{\Delta/s}}{q_{0,j}} } \\
            &   \leq \exp\bc{o\bc{\Delta}}   \exp\bc{-\frac{\Delta}{s} \sum_{j\in [s]} \KL{q_{1,j} }{q_{0,j}} } = \exp(-(1+o(1)) \Delta\cM). 
\end{align*} 
Together with \Cref{lem_calc_cont}, this establishes the claim as $\pr\bc{\cE^c} = o(n^{-2})$.
\end{proof}

\begin{proof}[Proof of \Cref{lem_nondef}]
Let $\vZ$ be the counting random variable on the l.h.s. of \eqref{eq_lem_nondef}.
Using \Cref{claim_pnondef} and \Cref{lem_calc}, as well as $\Delta = c d^* \log(n/k)$, 
    \begin{align} \label{eq_expnondef}
       \Erw\brk{\vZ} &= \sum_{i=s+1}^{\ell} \sum_{x \in V[i]} \pr\bc{x \in V_0[i]} \pr\bc{\bigcap_{j\in [s]} \bigcap_{r\in \{0,\dots,t-1\}}  \fE_{x,j,r} ~\Big\vert~ x\in V_0[i]} \nonumber \leq  n \exp\bc{-(1+o(1))\Delta\cM} \\
       & = n \exp\bc{-(1+o(1))c H(\Pr\bc{\Po(d^*) \geq t }) \log \frac{n}{k}}   \leq n \exp\bc{-(1+\eps)\log\frac{n}{k}} .
    \end{align}
 In the last step, we have used that $c \geq (1+\eps)\max\{c_1(d^*), c_2(d^*, \theta)\}$. By Markov's inequality,
    \begin{align*}
         \Pr\bc{\vZ \geq k^{1-\delta}}&\leq \frac{n}{k} \exp\bc{-(1+\eps)\log\frac{n}{k}} k^{\delta} 
        \leq \exp\bc{-\eps\log\frac{n}{k} } k^{\delta} \\
        & \leq \exp\bc{-\eps\log\frac{n}{k}+\delta \log{k}} = o(1) \ ,
    \end{align*}
    when $\delta \theta <\eps (1-\theta),$ since $k=n^{\theta}$ and $n/k=n^{1-\theta}$ with $\theta<1.$
\end{proof}

\subsection{Misclassified Defective Items (Proof of \Cref{lem_def})}\label{sec:lem_def}
\invisible{\lemdef*}
In this section, we prove \Cref{lem_def}. For this, we again rely on Claims \ref{claim_expextation} and \ref{claim_distpq} in the previous section:

\begin{proof}[Proof of \Cref{lem_def}]
Recall the event $\cE$ from \Cref{Lemma_GammaMinMax}. First, observe that $\vW_{x,s}^{r,\pm}(\SIGMA) =  W_{1,s}^{r,\pm} = 0$ for $r \in [t-1]$. Thus, the indicators in \eqref{eq_def_count} corresponding to these index choices are deterministically zero. In all other cases, as explicitly checked in \eqref{eq_lower_W} for $j \in [s-1]$, $\zeta W_{1,s}^{r,\pm} = \Omega(\zeta \frac{\Delta}{s^{t+1}})$.
By Claim \ref{claim_distpq}, on $\cE$,
    \begin{align*}
    \Big|\Erw\brk{\vW_{x,j}^{r,\pm}(\SIGMA)  \vert \SIGMA_x=1, \vk} - W_{1,j}^{r,\pm}\Big|\leq W_{1,j}^{r,\pm} O(\log^3 n/\sqrt{k})\leq W_{1,j}^{r,\pm} (\zeta/4).
    \end{align*}

Therefore, thanks to Claim \ref{claim_expextation}, and using the Chernoff bound for the Binomial distribution, we can estimate 
\begin{align} \label{eq_estimate_prob}
  &\pr\bc{\abs{\vW_{x,j}^{r,\pm}(\SIGMA) - W_{1,j}^{r,\pm}} > W_{1,j}^{r,\pm}(\zeta/2) ~\Big\vert~ \SIGMA_x=1, \cE, \vk } \nonumber \\ 
  &\leq \pr\bc{\abs{\vW_{x,j}^{r,\pm}(\SIGMA) - \Erw\brk{\vW_{x,j}^{r,\pm}(\SIGMA) \vert \SIGMA_x=1, \vk }} > W_{1,j}^{r,\pm}(\zeta/4) ~\Big\vert~ \SIGMA_x=1, \cE, \vk }  \nonumber \\
  &\leq \exp\bc{-\Omega\bc{\zeta^2 \frac{\Delta^2 s}{s^{2t+2} \Delta}} } = \exp\bc{-\Omega\bc{\zeta^2 \frac{\Delta}{s^{2t+1}}}}.
\end{align}
Again, let $\vZ$ denote the random variable on the l.h.s. of \eqref{eq_def_count}. A union bound together with \eqref{eq_estimate_prob} gives
\begin{align*}
  \Erw\brk{Z} \leq  k s t \cdot \exp\bc{-\Omega\bc{\zeta^2 \frac{\Delta}{s^{2t+1}}}} \leq k \cdot \exp\bc{-\Omega\bc{\frac{\log n}{\bc{\log \log n}^{2t+3}}}}.
\end{align*}
By Markov's inequality, w.h.p., the number of misclassified defective items is therefore at most $$k\cdot \exp\bc{-\Omega\bc{\frac{\log n}{\bc{\log \log n}^{2t+3}}}}. \qedhere$$   
\end{proof}

\section{Proof of \Cref{prop_endgame} (Cleaning Phase)} \label{Sec:Clean}
For each defective item $x$, we first define the subset of tests that will be {\em pivotal} in determining its true label by
\begin{align}
    \vS_x[i] = \abs{\cbc{a \in \partial x \cap F[i]\colon \abs{V_1 \cap \bc{\partial a \setminus \cbc{x}}} = t-1}} \quad \text{and} \quad \vS_x = \sum_{i=1}^\ell \vS_x[i].
\end{align}
Here, the notation $\partial a \setminus \cbc{x}$ indicates that all occurrences of $x$ in $a$ are disregarded.
The following lemma ensures that the number of pivotal tests is $\Omega(\sqrt \Delta)$:
\begin{lemma}[\Whp\ defective items have many pivotal tests]
\label{Lem_distphixstar}
Assume that $m\geq(1+\eps)\minf$.
Then $\min_{x\in\one}\vS_x\geq\sqrt \Delta$ \whp
\end{lemma}
\begin{proof}
Recall that $P_l(\mu)$ is the probability that a Poisson random variable with parameter $\mu$ takes the value $l$. Consider the experiment of first constructing the test design $\G_{\mathrm{sc}}$ and then re-sampling the neighbourhood $\partial x$; that is, independently of $\G_{\mathrm{sc}}$, the item $x$ joins $\Delta/s$ tests chosen uniformly at random from each compartment $F[i+j-1]$. The resulting design $\G'$ has the same distribution as $\G_{\mathrm{sc}}$. Consider the set
\begin{align*}
F_{t-1,x}[i+j-1]
=
\cbc{a \in F[i+j-1]:
\abs{V_1 \cap \bc{\partial a \setminus \cbc{x}}}=t-1}
\end{align*}
of tests in compartment $F[i+j-1]$ that contain exactly $t-1$ defective items other than $x$.
 Let $\vS_x'[i+j-1]$ denote the number of tests in $F[i+j-1]\cap\partial x$ that belong to $F_{t-1,x}[i+j-1]$. Then, conditioned on $\G$,
\begin{align}\label{eqeqLem_distphixstar3_new}
\vS_x'[i+j-1]
\sim
\Bin\bc{ \frac\Delta s,  \abs{F_{t-1,x}[i+j-1]}\frac{\ell}{m}\,
}.
\end{align}
Moreover, without conditioning on $\G_{\mathrm{sc}}$, $\vS_x[i+j-1] \sim \vS_x'[i+j-1]$. 

We next argue that conditioning on $\G_{\mathrm{sc}}$ provides only negligible information about $\vS_x'[i+j-1]$. 
Let $\vm_{t-1}^{(i)}$ denote the number of tests in compartment $i$ that have exactly $t-1$ edges connecting them to defective items.
The following lemma shows that $\vm_{t-1}^{(i)}$ is \whp\ close to $\tfrac{m}{\ell}P_{t-1}(d)$:

\begin{restatable}[The number of test with $t-1$ edges to defective items]{lemma}{lemma_m0_spatial}
\label{Lemma_m0_spatial}
Assume $m = ck \log(n/k)$ and $\Delta = cd \log(n/k)$ for some constants $c,d>0$. For any fixed $r \in \mathbb{N}$, with probability at least $1-o(n^{-1})$, for all $i \in [\ell]$,
        $$\vm_{r}^{(i)}\in \brk{\bc{1\pm n^{-\Omega(1)}}\frac{m}{\ell} P_r(d)}.$$
\end{restatable}

\begin{proof}
The proof of \Cref{Lemma_m0_spatial} is deferred to Appendix \ref{appendix-C}.
\end{proof}

Next, consider the \whp\ event (see \Cref{Lemma_m0_spatial})
\begin{align*}
    \cN = \cbc{\forall i \in [\ell] : \vm_{t-1}^{(i)}\in \brk{\bc{1\pm n^{-\Omega(1)}}\frac{m}{\ell}P_{t-1}(d)}}.
\end{align*}
Since $x$ participates in  $\Delta/s=O(\log n)$ tests per compartment, 
given $\cN$,
\begin{align}\label{eqLem_distphixstar2_new}
\abs{F_{t-1,x}[i+j-1]}
&=
\vm_{t-1}^{(i+j-1)}+O(\log n) =
\bc{1\pm n^{-\Omega(1)}}
\frac{m}{\ell} P_{t-1}(d).
\end{align}
Combining \eqref{eqLem_distphixstar2_new} and \eqref{eqeqLem_distphixstar3_new} yields
that given $\cN$, for each $x \in V_{1}[i]$, $i\in[\ell]$, the random variable $\vS_x$ has distribution
\begin{align}\label{eqLem_distphixstar1}
\vS_x[i+j-1]\sim \Bin\bc{\frac \Delta s,\bc{1\pm n^{-\Omega(1)}}P_{t-1}(d)}.
\end{align}

To complete the proof, we combine \eqref{eqLem_distphixstar1} with \Cref{lem_chernoff_2}, which implies that, \whp,
\begin{align}
\pr\bc{\vS_x[i+j-1]\leq \sqrt\Delta\mid x\in V_1}&\leq\exp\bc{-\frac \Delta s\KL{(1+o(1))s/\sqrt{\Delta}}{P_{t-1}(d)+o(1)}} \\
&=\exp\bc{-(1+o(1))\frac{\Delta\log \bcfr{1}{1-P_{t-1}(d)}}{s}}.\label{eqeqLem_distphixstar4}
\end{align}
Given $\G_{\mathrm{sc}}$,  the random variables $(\vS_x'[i+j-1])_{j\in[s]}$ are mutually independent. A similar argument as above in combination with 
\eqref{eqeqLem_distphixstar4} yields
\begin{align}\label{eqeqLem_distphixstar5}
\pr\bc{\vS_x\leq \sqrt\Delta\mid x\in V_1}&\leq\bcfr{1}{1-p_{t-1}}^{-(1+o(1))\Delta}.
\end{align}
Finally, the assumption $m\geq(1+\eps)\minf$ for a fixed $\eps>0$, combined with the choice \eqref{Eq:Param} of $\Delta$,
ensure that $(1/(1-p))^{-(1+o(1))\Delta}=o(1/k)$.
Thus, the assertion follows from \eqref{eqeqLem_distphixstar5} by taking a union bound on $x\in V_1$.
\end{proof}

We can now combine the fact that each defective item is included in an abundance of pivotal tests, with the expansion argument from \Cref{prop_endgame}, to show that the set of misclassified items rapidly decreases, thereby establishing the result for exact recovery:

\begin{proof}[Proof of \Cref{prop_endgame}]
For $j = 1, \ldots, \ceil{\log n}$, let
\begin{align*}
{\sf Mis}_j = \cbc{ x \in V: \tau^{(j)}_x \neq \SIGMA_x}
\end{align*}
denote the set of items that remain misclassified at the $j$th iteration of the clean-up procedure. An item can remain misclassified at iteration $j+1$ only if many of the tests it participates in are disrupted by other misclassified items from iteration $j$. We will therefore show that this forces ${\sf Mis}_{j+1}$ to be much smaller than ${\sf Mis}_j$. Using \Cref{lem_nondef} and \Cref{lem_def}, we find that, \whp,
\begin{align}\label{eqprop_endgame1}
\abs{{\sf Mis}_1} \leq k\exp\bc{-\frac{\log n}{(\log\log n)^{2t+3}}}.
\end{align}
Furthermore, in light of \Cref{Lem_distphixstar} we may condition on the event $\cA=\{\min_{x\in\one}\vS_x\geq\sqrt \Delta\}$.

We now show that given $\cA$ for every $j\geq1$, deterministically
\begin{align}\label{eqprop_endgame2}
{\sf Mis}_{j+1}\subset{\cbc{x\in V\colon \sum_{a\in\partial x\setminus F[0]}\vecone\{\partial a \cap {\sf Mis}_j\setminus\cbc x \not= \varnothing\} \geq \ceil{\log^{1/4} n}}}.
\end{align}
In simpler terms, an item remains misclassified at iteration $j+1$ only if it is involved in at least $\log^{1/4} n$ tests that contain other misclassified items.

To prove \eqref{eqprop_endgame2}, suppose that $x\in{\sf Mis}_{j+1}$. Recall from \eqref{Eq:Param} that $\Delta=\Omega(\log n)$ (since $c \geq (1+\eps)\cinf$).
Moreover, Step~15 of \SPOT\ classifies $x$ based upon the number $\vS_x\bc{\tau^{(j)}}$ of positive tests containing $x$ that contain exactly $t-1$ other items (with multiplicities) classified as defective, where
\begin{align*}
\vS_x\bc{\tau^{(j)}} = \sum_{a\in\partial x\colon \hat\SIGMA_a=1} \vecone\cbc{\abs{\cbc{y\in\partial a\setminus\cbc{x}\colon \tau^{(j)}_y=1}} = t-1}.
\end{align*}

There are two cases to consider:
\begin{description}
\item[Case 1 ($x\in V_0$) :]
Suppose that $\SIGMA_x=0$, but $\tau^{(j)}_x=1$. 
Since Step~15 of \SPOT\ applies the threshold $\vS_x\bc{\tau^{(j)}} > \log^{1/4}n$, there are at least $\log^{1/4}n$ tests $a\in\partial x$ such that
\begin{align*}
\abs{\cbc{y\in\partial a\setminus\cbc{x}:\tau^{(j)}_y=1}} = t-1.
\end{align*}
However, every positive test $a\in\partial x$ must contain at least $t$ actually defective items other than $x$. So every such test $a$ contains at least one truly defective items in $\partial a\setminus\cbc{x}$ that is labelled incorrectly as non-defective under $\tau^{(j)}$ and thus contained in ${\sf Mis}_j$, i.e.,
\begin{align*}
 \sum_{a\in\partial x\setminus F[0]}\vecone\{\partial a \cap {\sf Mis}_j\setminus\cbc x \not= \varnothing\} \geq \ceil{\log^{1/4} n}.
\end{align*}

\item[Case 2 ($x\in V_1$) :]
Suppose that $\SIGMA_x=1$, but $\tau^{(j)}_x=0$. Since Step~15 of \SPOT\ applies the threshold $\vS_x\bc{\tau^{(j)}} > \log^{1/4}n$, there are at most $\log^{1/4}n$ tests $a\in\partial x$ such that
\begin{align*}
\abs{\cbc{y\in\partial a\setminus\cbc{x}:\tau^{(j)}_y=1}} = t-1.
\end{align*}
On the other hand, on the event $\cA$, $x$ in fact participates in 
$\vS_x \geq \sqrt\Delta = \Omega(\log^{1/2}n)$ 
tests that contain exactly $t-1$ other defective items.
Therefore, there are at least
\begin{align*}
\sqrt\Delta - \log^{1/4}n = \Omega(\log^{1/2}n)
\end{align*}
tests $a\in\partial x$ that must contain a misclassified item other than $x$ at step $j$. Hence,
\begin{align*}
 \sum_{a\in\partial x\setminus F[0]}\vecone\{\partial a \cap {\sf Mis}_j\setminus\cbc x \not= \varnothing\} \geq \ceil{\log^{1/4} n}.
\end{align*}
\end{description}
In both cases, we obtain \eqref{eqprop_endgame2}.
Finally, combining \eqref{eqprop_endgame1}, \eqref{eqprop_endgame2} and the expansion property from \Cref{lemma_endgame_misclassified} we conclude that \whp\ $|{\sf Mis}_{j+1}|\leq|{\sf Mis}_{j}|/3$ for all $j\geq1$.
Iterating this bound for $\lceil \log n \rceil$ steps yields ${\sf Mis}_{\lceil\log n\rceil}=\varnothing$ \whp\ 
\end{proof}

\begin{funding} \phantom{a}
\vspace{-0.8cm}
\InsertBoxR{0}{\includegraphics[scale=0.1]{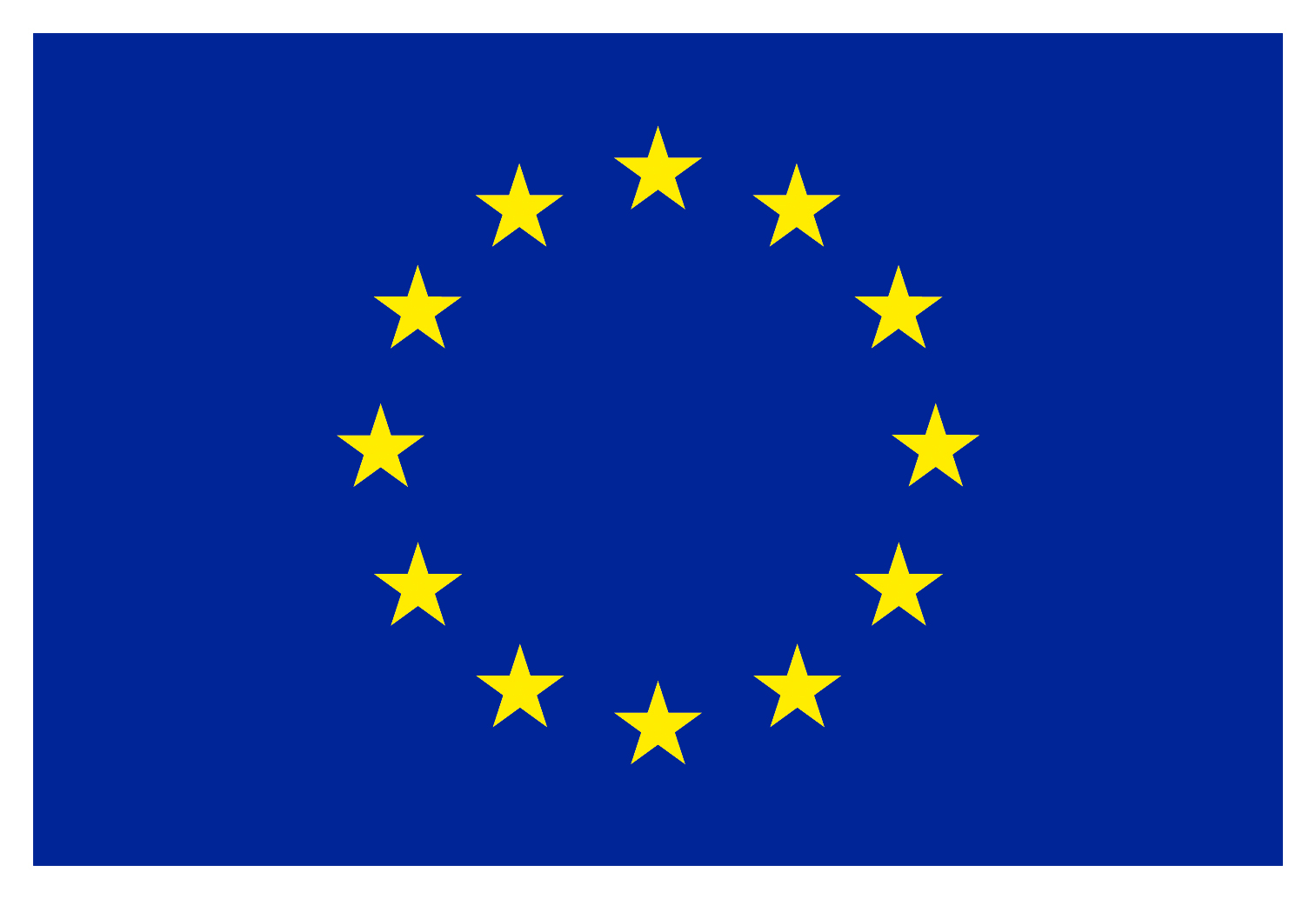}}
\noindent 
Amin Coja-Oghlan is supported by DFG CO 646/3, DFG CO 646/5 and DFG CO 646/6. 
Remco van der Hofstad and Noela Müller are supported in part by the NWO Gravitation project NETWORKS under grant no.\ 024.002.003.
Lena Krieg is supported by DFG CO 646/3.
\end{funding}



\bibliographystyle{imsart-number} 
\bibliography{bibliography_j}       


\begin{appendix}
\section{Probabilistic Toolbox} \label{Appendix_inequalities}
In this section, we state some results that we use in the main proofs.

\begin{theorem}[Chernoff bound for the binomial distribution \cite{janson2011random}] 
\label{lem_chernoff}
Let $\vX \sim \Bin \bc{n, p}$. Then, for any $\eps>0$, 
\begin{align*}
    & \Pr\bc{ \vX \geq (1 + \eps) \Erw \brk{\vX}} \leq \exp\bc{-\frac{\eps^2}{2 + 2\eps/3} \Erw\brk{\vX}} \qquad \text{and}  \\
    & \Pr\bc{\vX \leq (1 - \eps) \Erw\brk{\vX}} \leq \exp\bc{-\frac{\eps^2}{2} \Erw\brk{\vX}}.
\end{align*}
\end{theorem}

We will  occasionally apply the following Chernoff bound.

\begin{theorem}[Chernoff bound for the binomial distribution alternative \cite{janson2011random}] \label{lem_chernoff_2}
Let $\vX \sim \Bin \bc{N, p}$. Then
\begin{align}\label{eqChernoff1}
    \Pr\bc{\vX \geq {qN}} &\leq \exp \bc{-N\KL{q}{p}} \quad \text{for $p<q<1$,} \\
    \Pr\bc{\vX \leq {qN}} &\leq \exp \bc{-N\KL{q}{p}} \quad \text{for $0<q<p$.}\label{eqChernoff2}
\end{align}
\end{theorem}

\begin{theorem}[Chernoff bound for the hypergeometric distribution \citep{Janson99onconcentration}]\label{Hyp}
Let $\vX$ be a Hyp$(N,M,K)$-distributed random variable and let $t \geq 0$. Then
\begin{align*}
    & \Pr\bc{\vX - \Erw[\vX] \geq t} \leq \exp\bc{-\frac{t^2}{2(KM/N+t/3)}} \qquad \text{and} \\ & \Pr\bc{\vX - \Erw[\vX] \leq -t} \leq \exp\bc{-\frac{t^2}{2KM/N}}.
\end{align*}
\end{theorem}

\begin{theorem}[Poisson Approximation with error terms \cite{barbour1992poisson}]\label{PoiApprox}
Let $\vX \sim \Bin(n,p)$, then for each $k \in \mathbb{Z}^+$ 
\[ \Pr(\vX=k) = \binom{n}{k}p^k (1-p)^{n-k} = \frac{(np)^k}{k!} e^{-np} \bc{1+O\bc{np^2,k^2 n^{-1}}}. \]
\end{theorem}

\begin{definition}[Negative association \cite{dubhashi1998balls}] \label{def_NA} Let $\vX = (\vX_1, \ldots, \vX_d)$ be a random vector. Then the family of random variables $\vX_1, \ldots, \vX_d$ is said to be {\em negatively associated} if, for every two disjoint index sets $I,J \subseteq [d]$,
\begin{align*}
    \Erw\brk{f(\vX_i: i \in I) g(\vX_j: j \in J)} \leq \Erw\brk{f(\vX_i: i \in I)} \Erw\brk{ g(\vX_j: j \in J)}
\end{align*}
for all functions $f:\RR^{|I|} \to \RR$ and $g:\RR^{|J|} \to \mathbb{R}$ that are either both non-decreasing or both non-increasing.
\end{definition}

\begin{theorem}[Chernoff Bound under negative association \cite{dubhashi1998balls}] \label{ChernNeg}
    The Chernoff-Hoeffding bound Theorem~\ref{lem_chernoff} also applies to sums of random variables that satisfy the negative association condition from Definition \ref{def_NA}.
\end{theorem}

\section{Proof of \Cref{claim_distpq}} \label{Appx:disqp}
\begin{proof}[Proof of \Cref{claim_distpq}] 
Recall that $\mathcal{E}$ denotes the event that the statements 1.-3. from \Cref{Lemma_GammaMinMax} hold. 
Fix $i \in \{s+1,\ldots, \ell\}$ and $j \in [s]$. For abbreviation let $\hat{\vk}=\hat{\vk}_{j,i} = \vec k[i-(s-j);i-1]$. Given $\cE$, for all $i,j$, $\hat{\vk}\ell \in \brk{(s-j)k \pm \lambda}$, where $\lambda = \sqrt{kl} (s-j) \log n$. We will first calculate the difference between the point probabilities of the Binomial and the Poisson distributions that occur in the definitions of $\vec p$ and $q$. 
We prove upper and lower bounds. For $n$ large enough, and starting with the upper bound,
\allowdisplaybreaks
\begin{align*}
&\frac{\Pr\bc{\Bin(\vec k[i-(s-j);i-1] \Delta/s,\frac{\ell}{m}) = r \mid \vk, \mathcal{E}}}{P_r(d^*(s-j)/s) } 
= \frac{\frac{(\hat{\vk} \Delta/s)!}{r! (\hat{\vk} \Delta/s - r)!} \bc{\frac{\ell}{m}}^r\bc{1- \frac{\ell}{m}}^{\hat{\vk}\Delta/s - r} }{\exp(-d^*(s-j)/s) \frac{(d^*(s-j)/s)^r}{r!}} \\
&\leq \bc{\frac{\hat{\vk}\ell \Delta}{m d^*(s-j)}}^r\bc{1- \frac{\ell}{m}}^{\hat{\vk}\Delta/s - r} \exp(d^*(s-j)/s) \\
&\stackrel{\eqref{Eq:Param}}{\leq} \bc{\frac{\hat{\vk}\ell}{k(s-j)}}^r\bc{1- \frac{\ell}{m}}^{\hat{\vk}\Delta/s - r} \exp(d^*(s-j)/s) \\
&\leq \bc{1+ \frac{\sqrt{\ell}}{\sqrt{k}} \log n}^r\exp\bc{-\frac{\ell(\hat{\vk}\Delta/s - r)}{m}} \exp(d^*(s-j)/s) \\
&\leq \exp\bc{-\frac{\ell(\hat{\vk}\Delta/s - r)}{m} + d^*(s-j)/s + r \frac{\sqrt{\ell}}{\sqrt{k}} \log n} \\
&\leq \exp\bc{r \bc{\frac{\ell}{m} + \frac{\sqrt{\ell}}{\sqrt{k}} \log n} + \frac{\lambda}{sk}} \leq \exp\bc{\frac{\log^2 n}{\sqrt{k}}} \leq 1 + O\bc{\frac{\log^2 n}{\sqrt{k}}}.
\end{align*}
Similarly, for the lower bound,
\allowdisplaybreaks
\begin{align*}
&\frac{\Pr\bc{\Bin(\vec k[i-(s-j);i-1] \Delta/s,\frac{\ell}{m}) = r\mid \vk, \cE}}{P_r(d^*(s-j)/s) } \\ 
&= \frac{(\hat{\vk} \Delta/s)!}{ (\hat{\vk} \Delta/s - r)!} \bc{\frac{\ell}{m (d^*(s-j)/s)}}^r\bc{1- \frac{\ell}{m}}^{\hat{\vk}\Delta/s - r} \exp(d^*(s-j)/s) \\
&\geq \bc{\hat{\vk}\frac{\Delta}{s}-r}^r \bc{\frac{\ell}{m (d^*(s-j)/s)}}^r\bc{1- \frac{\ell}{m}}^{\hat{\vk}\Delta/s - r} \exp(d^*(s-j)/s) \\
&\geq \bc{\frac{\hat{\vk}\ell}{k ((s-j))}-\frac{r\ell s}{m d^* (s-j)}}^r\bc{1- \frac{\ell}{m}}^{\hat{\vk}\Delta/s - r} \exp(d^*(s-j)/s) \\
&\geq \bc{1+ \frac{\sqrt{\ell}}{\sqrt{k}} \log n -\frac{r\ell s}{m d^* (s-j)}}^r \exp\bc{-\frac{\ell}{m}\bc{\hat{\vk}\Delta/s -r} + O\bc{\hat{\vk} \frac{d^* \ell^2}{k m s}}}\exp(d^*(s-j)/s) \\
&\geq \bc{1+ \frac{\sqrt{\ell}}{\sqrt{k}} \log n -O\bc{\frac{\ell s}{k \log n}}}^r \exp\bc{-\bc{\lambda\frac{d^*}{sk} } + O\bc{ \frac{(k(s-j)\pm \lambda)d^* \ell}{k m s}}} \\
&\geq \exp\bc{\frac{r \sqrt{\ell}}{\sqrt{k}} \log n + O\bc{\frac{\ell}{k} \log^2 n}} \exp\bc{O\bc{\frac{\log^2 n}{\sqrt{k}}}} \\
&\geq \exp\bc{O\bc{\frac{\log^2 n}{\sqrt{k}}}}^2 =\bc{1 + O\bc{\frac{\log^2 n}{\sqrt{k}}}}^2 = 1 + O\bc{\frac{\log^2 n}{\sqrt{k}}}. 
\end{align*}
We conclude that
\begin{align*}
    \frac{\Pr\bc{\Bin(\vec k[i-(s-j);i-1] \Delta/s,\frac{\ell}{m}) = r\mid \vk, \cE}}{P_r(d^*(s-j)/s)} = 1 + O\bc{\frac{\log^2 n}{\sqrt{k}}}. 
\end{align*}
An analogous calculation can be conducted to show that
\begin{equation*}
    \frac{\Pr\bc{\Bin(\vec k[i;i+j-1] \Delta/s,\frac{\ell}{m}) = r\mid \vk, \cE}}{P_r(d^*j/s) } = 1+O\bc{\frac{\log^2 n}{\sqrt{k}}} \, .
\end{equation*}
Additionally, with an eye on the quantities $\vec p_{1,j,i}^{r,\pm}$, for any fixed $t'$,
\begin{align*}
    \frac{\Pr\bc{\Bin(\vec k[i;i+j-1] \Delta/s,\frac{\ell}{m}) = t'\mid \vk, \cE}}{\Pr\bc{\Bin(\vec k[i;i+j-1] \Delta/s -1,\frac{\ell}{m}) = t'\mid \vk, \cE} } = 1 + O(t'/k),
\end{align*}
which gives
\begin{align*}
    \frac{\Pr\bc{\Bin(\vec k[i-(s-j);i-1] \Delta/s - 1,\frac{\ell}{m}) = r\mid \vk, \cE}}{P_r(d^*(s-j)) } = 1 + O\bc{\frac{\log^2 n}{\sqrt{k}}},
\end{align*}
and
\begin{equation*}
    \frac{\Pr\bc{\Bin(\vec k[i;i+j-1] \Delta/s - 1,\frac{\ell}{m}) = r\mid \vk,\cE}}{P_r(d^*j/s) } = 1+O\bc{\frac{\log^2 n}{\sqrt{k}}} \, .
\end{equation*}

Now that we know the difference between the point probabilities, we next determine the difference between the cumulative distribution function and the tail distribution function of the two distributions. For some constant $t'$, we compute
\begin{align*}
&  \frac{\Pr\bc{\Bin(\vec k[i;i+j-1] \Delta/s,\frac{\ell}{m}) \geq t' \mid \vk, \cE}}{P_{\geq t'}(d^*j/s) } = \frac{1-\Pr\bc{\Bin(\vec k[i;i+j-1] \Delta/s,\frac{\ell}{m}) < t' \mid \vk, \cE}}{ 1 - P_{<t'}(d^*j/s) } \\
 =& \frac{1-\bc{1+O\bc{\frac{\log^2 n}{\sqrt{k}}}}P_{<t'}(d^*j/s) }{ 1 - P_{<t'}(d^*j/s) } 
 = 1 + \frac{O\bc{\frac{\log^2 n}{\sqrt{k}}}P_{<t'}(d^*j/s)}{ 1 - P_{<t'}(d^*j/s) } \\
 =& 1 + O\bc{\frac{\log^2 n}{\sqrt{k}}}\frac{1}{ 1 - P_{<t'}(d^*j/s) } 
 = 1 + O\bc{\frac{\log^2 n}{\sqrt{k}}}\frac{1}{ \sum_{i=t'}^\infty P_i(d^*j/s) } \\
 =& 1 + O\bc{\frac{\log^2 n}{\sqrt{k}}}\frac{1}{ \Theta(1) \bc{\frac{d^*}{s}}^{t'} } = 1 + O\bc{\frac{\log^3 n}{\sqrt{k}}}.
\end{align*}
In the penultimate step, we have used that $P_{\geq t'}(d^*j/s) \geq P_{t'}(d^*j/s) \geq \eul^{-d^*} \bc{\frac{d^*}{s}}^{t'}/(t')!$ for all $j \in [s]$ and that $d^*/s$ is of order $O(1/\log\log n)$.
On the other hand,
\begin{align*}
\frac{\Pr\bc{\Bin(\vec k[i;i+j-1] \Delta/s,\frac{\ell}{m}) \leq t' \vert \vk, \cE}}{\Pr\bc{\Po(d^*j/s) \leq t'} } &= \frac{\sum_{i=0}^{t'} \Pr\bc{\Bin(\vec k[i;i+j-1] \Delta/s,\frac{\ell}{m}) = i \vert \vk, \cE}}{\sum_{i=0}^{t'} \Pr\bc{\Po(d^*j/s) = i}} \\
&= \frac{\bc{1+O\bc{\frac{\log^2 n}{\sqrt{k}}}}\sum_{i=0}^{t'} \Pr\bc{\Po(d^*j/s) = i}}{\sum_{i=0}^{t'} \Pr\bc{\Po(d^*j/s) = i}} \\
&= 1+O\bc{\frac{\log^2 n}{\sqrt{k}}}.
\end{align*}
Combining the calculations gives all cases in Claim~\ref{claim_distpq}.
\end{proof}

\section{The number of pivotal tests (Proof of \Cref{Lemma_m0_spatial})} 
\label{appendix-C}
In this section, we investigate the number of pivotal tests, and prove \Cref{Lemma_m0_spatial}.
We recall that, in the cleaning phase, we require that the number of tests that have exactly $t-1$ edges connecting to defective items is well concentrated in {\em every} compartment.
Let $\vm_{r}^{(i)}$ denote the number of tests in compartment $i$ with exactly $r$ defectives. Recall that \Cref{Lemma_m0_spatial} investigates concentration properties of $\vm_{r}^{(i)}$ for different $i$ and $r$.

We start by adapting the setting slightly, so as to increase the level of independence. Indeed, recall that, in the spatially coupled test design, the numbers of of defective items in distinct tests are correlated. Following the approach of \cite[Appendix B]{aco_2019}, we work with a family of {\em independent} random variables in the proof, as we explain first.
 
Let  $\vGamma = \bc{\vGamma_j}_{j \in [m]}$ be the test degree sequence, which are weakly-dependent hypergeometric random variables. We define $(\vX_j)_{j \in [m]}$ to be a sequence of {\em independent binomial} random variables, where $\vX_j\sim\Bin(\vGamma_j, k/n)$. In addition, introduce the event 
\begin{align*}
	\cF&=\cbc{\sum_{j\in[m]}\vX_j=k\Delta}.
\end{align*}
By mutual independence of the $\vX_j$'s, and the local limit theorem for the binomial distribution, we obtain that
\begin{align}\label{eqEprob}
	\pr\bc{\cF \mid \vGamma} = \pr\bc{\Bin\bc{n\Delta,k/n}=k\Delta} =\Omega(1/\sqrt{\Delta k}).
\end{align}
Further, let $\vY_j$ denote the number of edges in test $a_j$ that are incident to defective items. 

\begin{lemma}[{\cite[Lemma B.2]{aco_2019}}]\label{Lemma_Elemma}
    Let $(\vY_1, \dots, \vY_m)$ and $(\vX_1, \dots, \vX_m)$ be as defined above. Then, for any integer sequence $(y_j)_{j\in[m]}$ with $0\leq y_j \leq \vGamma_j$ and $\sum_{j\in[m]}y_j=k\Delta$, 
$$
\Pr\Big(\forall j \in [m]\colon \vY_j = y_j \;\big|\; \vGamma\Big)
=
\Pr\Big(\forall j \in [m]\colon  \vX_j = y_j \;\big|\; \vGamma, \mathcal{F}\Big).
$$
\end{lemma}

\begin{proof}[Proof of Lemma~\ref{Lemma_m0_spatial}]
Fix a compartment $j\in[\ell]$. Observe that to estimate $\vm_{t-1}^{(j)}$, it is sufficient to consider the sub-graph $(V[j-s+1]\cup \ldots \cup V[j], F[j])$. This sub-graph forms a constant-column design with item-degrees $\tilde \Delta = \Delta/s$ and $\tilde m=m/\ell$ tests. The only small complication is that the total number of defectives in this sub-graph is {\em random}.

Let $\vk$ be the vector of the numbers of defective items per compartment (see \eqref{def_vk}). We condition on $\tilde \vk:= \vk[j-s+1:j]$ throughout. As above, for $a \in F[j]$, let $\vY_a$ be the number of edges of $a$ that connect to a defective item. Let $(\vX_a)_{a \in F[j]}$ be a sequence of independent random variables, where $\vX_a\sim\Bin(\vGamma_a, \tilde \vk/(sn/\ell))$. In addition, introduce the event
\begin{align*}
	\cF_{j}&=\cbc{\sum_{a\in F[j]}\vX_a=\tilde \vk\Delta/s}.
\end{align*}
As in \eqref{eqEprob}, by mutual conditional independence of the $\vX_j$'s, and the local limit theorem for the binomial distribution, we obtain that for $\tilde\vk = O\bc{sk/\ell}$,
\begin{align}\label{eqEprob_j}
	\pr\bc{\cF_j \mid \vGamma, \tilde\vk} = \pr\bc{\Bin\bc{\frac{n\Delta}{\ell}, \tilde \vk/(sn/\ell) }=\tilde\vk\Delta/s} =\Omega\bc{\sqrt{\frac{\ell}{\Delta k}}}.
\end{align}
By Lemma~\ref{Lemma_Elemma}, conditioned on $(\vGamma_a)_{a \in F[j]}$ and $\tilde \vk$, the vector $(\vY_a)_{a\in F[j]}$ has the same distribution as $(\vX_a)_{a\in F[j]}$ conditioned on  $\cF_{j}$.
Moreover,
\begin{align*}
\vm_{r}^{(j)} &= \sum_{a\in F[j]} \vecone\cbc{\vY_a=r},\qquad \text{and set} \quad \vm_{r}^{(j)\prime}= \sum_{a\in F[j]} \vecone\cbc{\vX_a=r}.
\end{align*}
By the distributional equivalence discussed above,
$$ \vm_{r}^{(j)}\mid (\tilde\vk, (\vGamma_a)_{a \in F[j]}) \stackrel{d}{=} \vm_{r}^{(j)\prime}\mid (\tilde \vk,(\vGamma_a)_{a \in F[j]},\cF_{j}). $$
Let $\vGamma$ satisfy \eqref{eq_comp_minmax_G},\eqref{eq_minmax_G} and $\tilde\vk$ satisfy 
$$ \tilde\vk = \bc{1+n^{-\Omega(1)}}\frac{ks}{\ell}. $$
By \Cref{Lemma_GammaMinMax} these properties, hold with probability $1-o(n^{-2})$.
For such $\tilde\vk, \vGamma$,
\begin{align} 
\pr\bc{\vX_a=r \mid \tilde\vk,\vGamma} &= \binom{\vGamma_a}{r} \bcfr{k}{n}^{r} \bc{1-k/n}^{\vGamma_a-r}
\bc{1+n^{-\Omega(1)}}.
\label{Eq:LocalBinApprox}
\end{align} 
Recall the shorthand $P_r(d) = \frac{d^{r}}{r!}\exp(-d)$. Again for $\vGamma$, $\tilde \vk$ as above, 
	\begin{align}
		\Erw[\vm_{r}^{(j) \prime} \mid \vGamma, \tilde \vk]&=\sum_{a \in F[j]} \binom{\vGamma_a}{r} \bcfr{\tilde \vk \ell}{ns}^{r}\bc{1-\frac{\tilde\vk \ell}{ns}}^{\vGamma_a-r} \leq \frac{m}{\ell} \binom{\vGamma_{\max}^{j}}{r}\bcfr{\tilde \vk \ell}{ns}^{r}\bc{1-\frac{\tilde\vk \ell}{ns}}^{\vGamma_{\min}^{j}-r} \nonumber\\
        &= \bc{1+n^{-\Omega(1)} } \frac{m}{\ell}  P_r(d) . \label{Eq:ExpW1U}	\end{align}
	Analogously,
	\begin{align}
		\Erw[\vm_{r}^{(j) \prime} \mid \vGamma, \tilde \vk ]&\geq \frac{m}{\ell}\binom{\vGamma_{\min}^{j}}{r}\bcfr{\tilde \vk \ell}{ns}^{r}\bc{1-\frac{\tilde\vk \ell}{ns}}^{\vGamma_{\max}^{j}-r} 
		=\bc{1+n^{-\Omega(1)}}  \frac{m}{\ell}  P_r(d) . \label{Eq:ExpW1L}
	\end{align}
    We conclude from \eqref{Eq:ExpW1U} and \eqref{Eq:ExpW1L} that, on the event $\cE$,
$$
\Erw[\vm_{r}^{(j) \prime} \mid \vGamma, \tilde \vk ]
= \bc{1+n^{-\Omega(1)}} \bc{1+n^{-\Omega(1)}}  \frac{m}{\ell}  P_r(d) .
$$
        The Chernoff bound on the Poisson-binomial random variable $\vm_{r}^{(j) \prime}$ then gives, again on the event $\cE$, 
	\begin{align}\label{Eq:PrW1}
	    \pr\bc{\left.\vm_{r}^{(j) \prime} \in \brk{\bc{1\pm n^{-\Omega(1)}}\frac{m}{\ell}P_r(d) }~\right|~\vGamma, \tilde\vk}=1-o(n^{-9}).
	\end{align}
    We now use the fact that $\vm_{r}^{(j) \prime}$, given $ \vGamma, \tilde\vk$ and $\cF_{j}$, is equal in distribution to $\vm_{r}^{(j)}$ given $\vGamma$ and $\vk$, to lower bound the probability that $\vm_{r}^{(j)} $ is close to its expectation as
    \begin{align}
        \pr&\bc{\vm_{r}^{(j)}\notin \left. \brk{\bc{1\pm n^{-\Omega(1)}}\frac{m}{\ell}P_r(d)  }~\right|~\vGamma, \tilde\vk} 
        = \pr\bc{ \vm_{r}^{(j) \prime}\notin \left. \brk{\bc{1\pm n^{-\Omega(1)}} \frac{m}{\ell}P_r(d) }~\right|~\vGamma, \tilde\vk, \cF_{j}} \nonumber \\ 
        &= \frac{\pr\bc{\vm_{r}^{(j) \prime} \not\in \left.\brk{\bc{1 \pm n^{-\Omega(1)}}\frac{m}{\ell}P_r(d)  }, \cF_{j}~\right|~ \vGamma, \tilde\vk}}{\Pr\bc{\cF_{j} \mid \vGamma, \tilde\vk}} \notag \\
        & \leq  \frac{\pr\bc{\vm_{r}^{(j) \prime} \not\in \left. \brk{\bc{1 \pm n^{-\Omega(1)}}\frac{m}{\ell}P_r(d)  }~\right|~\vGamma, \tilde\vk}}{\Pr\bc{\cF_{j} \mid \vGamma, \tilde\vk}} = o\bc{n^{-7}}, \label{Eq:CloseExp}
    \end{align}
   with probability $1-o(n^{-7})$, which follows from \eqref{Eq:PrW1} and \eqref{eqEprob_j}. Taking a union bound over all compartments now gives the claim.
\end{proof}

\end{appendix}

\end{document}